\documentclass[twocolumn,10pt]{IEEEtran}
\usepackage{ifpdf, flushend}

\ifCLASSINFOpdf
  \usepackage[pdftex]{graphicx}
  \graphicspath{{../pdf/}{../jpeg/}}
  \DeclareGraphicsExtensions{.pdf,.jpeg,.png}
\else
  \usepackage[dvips]{graphicx}
  \graphicspath{{../eps/}}
  \DeclareGraphicsExtensions{.eps}
\fi
\usepackage[cmex10]{amsmath}
\usepackage {amssymb}
\usepackage{algorithmic}
\usepackage{array}
\usepackage{mdwmath}
\usepackage{mdwtab}
\usepackage{eqparbox}
\usepackage{url}
\usepackage{hyperref}
\usepackage{algorithm}
\usepackage{algorithmic}
\usepackage{epstopdf,cite,color}
\newtheorem{theorem}{Theorem}
\newtheorem{lemma}{Lemma}

\renewcommand{\baselinestretch}{1}

\newcommand{\beq}{\begin{equation}}
\newcommand{\eeq}{\end{equation}}
\newcommand{\beqn}{\begin{eqnarray}}
\newcommand{\eeqn}{\end{eqnarray}}
\newcommand{\beqno}{\begin{eqnarray*}}
\newcommand{\eeqno}{\end{eqnarray*}}
\newcommand{\bma}{\begin{displaymath}}
\newcommand{\ema}{\end{displaymath}}
\newcommand{\bnu}{\begin{enumerate}}
\newcommand{\enu}{\end{enumerate}}
\newcommand{\bce}{\begin{center}}
\newcommand{\ece}{\end{center}}
\newcommand{\btb}{\begin{tabular}}
\newcommand{\etb}{\end{tabular}}

\begin{document}

%
\title{Power Allocation for Full-Duplex Relay Selection in Underlay Cognitive Radio Networks: Coherent versus Non-Coherent Scenarios}

\author{Le~Thanh~Tan,~\IEEEmembership{Member,~IEEE,} Lei~Ying,~\IEEEmembership{Member,~IEEE,} and Daniel W. Bliss,~\IEEEmembership{Fellow,~IEEE}
\thanks{This paper will be presented in part at IEEE CISS 2017.}
\thanks{The authors are with the School of Electrical, Computer and Energy Engineering, Arizona State University,  Tempe, AZ 85287, USA. 
Emails: \{tlethanh,lei.ying.2,d.w.bliss\}@asu.edu. \textbf{L.~ T.~ Tan} is the corresponding author.}}


\maketitle

\begin{abstract}
This paper investigates power control and relay selection in Full Duplex Cognitive Relay Networks (FDCRNs), where the secondary-user (SU) relays can simultaneously receive data from the SU source and forward them to the SU destination. 
We study both non-coherent and coherent scenarios. 
In the non-coherent case, the SU relay forwards the signal from the SU source without regulating the phase; while in the coherent scenario, the SU relay regulates the phase when forwarding the signal to minimize the interference at the primary-user (PU) receiver. 
We consider the problem of maximizing the transmission rate from the SU source to the SU destination subject to the interference constraint at the PU receiver and power constraints at both the SU source and SU relay. 
We then develop a mathematical model to analyze the data rate performance of the FDCRN considering the self-interference effects at
the FD relay.
We develop low-complexity and high-performance joint power control and relay selection algorithms. 
Extensive numerical results are presented to illustrate the impacts of power level parameters and the self-interference cancellation quality on the rate performance.
Moreover, we demonstrate the significant gain of phase regulation at the SU relay.

\end{abstract}

\IEEEpeerreviewmaketitle

\section{Introduction}

Cognitive radio is one of the most promising technologies for addressing today's spectrum shortage \cite{Tan11,Liang11,Zou12}. 
This paper considers {\em underlay} cognitive radio networks where primary and secondary networks transmit simultaneously over the same spectrum under the constraint that the interference caused by the secondary network to the primary network is below a pre-specified threshold \cite{Kim081,Le08}.
Because one critical requirement for the cognitive access design is that transmissions on the licensed frequency band from PUs should  be satisfactorily protected from the SUs' spectrum access.
Therefore, power allocation for SUs should be carefully performed to meet stringent interference requirements in this spectrum sharing model.

In particular, we consider a cognitive relay network where the use of SU relay can significantly increase the transmission rate because of path loss reduction. Most existing research on underlay cognitive radio networks has focused on the design and analysis of cognitive relay networks with half-duplex (HD) relays  (e.g., see \cite{Tan11,Liang11,Zou12} and the references therein). 
Due to the HD constraint, SUs typically require additional resources where the SU relays receive and transmit data on different time slots or on orthogonal channels.
Different from these existing work, this paper considers full-duplex relays, which can transmit and receive simultaneously on the same frequency band \cite{Duarte12, Everett14, Sabharwal14, Day12, Day12b, Bliss07, Korpi14, Choi10, Jain11}. Comparing with HD relays, FD relays can achieve both higher throughput and lower latency with the same amount of spectrum.

Design and analysis of FDCRNs, however, are very different from HDCRNs due to the presence of {\em self-interference}, resulted from the power leakage from the transmitter to the receiver of a FD transceiver \cite{Everett14, Sabharwal14, Korpi14}. 
This self-interference may significantly degrade the communication performance of FDCRNs.
Hence this paper focuses on power control and relay selection in FDCRNs with explicit consideration of self-interference. We assume SU relays use the amplify-and-forward (AF) protocol, and further assume full channel state information in both the non-coherent and coherent scenarios and the transmit phase information in the coherent scenario.
This can be done by using the conventional channel estimation techniques \cite{Arslan07} and implicit/explicit feedback techniques \cite{Mudumbai09} which are beyond the scope of our work.
The contributions of this paper are summarized below.

~ 1) We first consider the power control problem in the non-coherent scenario. We formulate the rate maximization problem where the objective is the transmission rate from the SU source to the SU destination, the constraints include the power constraints at the SU source and SU relay and the interference constraint at the PU receiver, and the optimization variables are the transmit power at the SU source and SU relay. The rate maximization problem is a non-convex optimization problem. However, it becomes convex if we fix one of two optimization variables (i.e., fixing the transmit power at the SU source or the transmit power at the SU relay). Therefore, we propose an alternative optimization algorithm to solve the power control problem. After calculating the achievable rate for each FD relay using the alternative optimization algorithm, the algorithm selects the one with the maximum rate.

~ 2) We then consider the coherent scenario, where in addition to control the transmit power, a SU relay further regulates the phase of the transmitted signal to minimize the interference at the PU receiver. We also formulate a rate maximization problem, which again is nonconvex. For this coherent scenario, we first calculate the phase to minimize the interference at the PU receiver. Then we prove that the power-control problem becomes convex when we fix either the transmit power of the SU source then optimize the transmit power of the SU relay or vice versa. We then propose an alternative optimization method for power control. Based on the achievable rate calculated from the alternative optimization algorithm for each relay, the relay with the maximum rate is selected.

~ 3) Extensive numerical results are presented to investigate the impacts of different parameters on the SU network rate performance and the performance of the proposed power control and relay selection algorithms. From the numerical study, we observe significant rate improvement of FDCRNs compared with HDCRNs. Furthermore, the coherent mechanism yields significantly higher throughput than that under the non-coherent mechanism.

\subsection{Related Work}

The FD technology can improve  spectrum access efficiency in cognitive radio networks \cite{Cheng14, Tan15a, Tan15, Tan16} where SUs can sense and transmit simultaneously. 
\cite{Cheng14} developed an FD MAC protocol that allows simultaneous spectrum access of the SU and PU networks
where both PUs and SUs are assumed to employ the MAC protocol for channel contention resolution and access.
This design is, therefore, not applicable to the hierarchical spectrum access in the CRNs where PUs should have higher spectrum access priority compared to SUs.
In \cite{Tan15a}, the authors propose the FD MAC protocol by using the standard backoff mechanism as
in the 802.11 MAC protocol where the system allows concurrent FD sensing and access during data transmission as well as frame fragmentation.
This design has also been studied in \cite{Tan15, Tan16} for the single- and multi-channel scenarios, respectively.
However all of these existing results assume the interweave spectrum sharing paradigm under which SUs only transmit when PUs are not transmitting.

Moreover, engineering of a cognitive FD relaying network has been considered in \cite{Kim12, Kim15}, where various resource allocation algorithms to improve the outage probability have been proposed.
In \cite{Zhong15}, the authors developed a mathematical model to analyze the outage probability for the proposed FD relay-selection scheme over both independent Nakagami-$m$ and Rayleigh fading channels.
The authors also extended their analysis to two-way FD-based AF relays in underlay cognitive networks \cite{Zhong16} where they analyzed various performance metrics such as outage probability, symbol error probability, etc.
In addition, \cite{Ramirez13} proposed a joint routing and distributed resource allocation for FD wireless networks.
\cite{Choi15} investigated distributed power allocation for a hybrid FD/HD system where all network nodes operate in the HD mode except  the access point (AP).
These existing results focus on either minimizing the outage probability or analyzing performance for existing algorithms.
This paper considers both non-coherent and coherent FU relay nodes and focuses on maximizing SU throughput given interference constraint at the PU and power constraints.

The remaining of this paper is organized as follows. Section ~\ref{SystemModel} describes the system model. Section ~\ref{PAFRS_Problem} studies the dynamic power allocation policies for full-duplex cognitive relay selection to achieve the maximum SU network rate.
Sections ~\ref{PCRS_Configuration_NonCoh} and \ref{PCRS_Configuration_Coh} consider the non-coherent and coherent scenarios, respectively.
Section ~\ref{Results} demonstrates numerical results followed by concluding remarks in Section ~\ref{conclusion}.
A part of this work will be presented at CISS 2017 \cite{Tan17}.

\section{System Models}
\label{SystemModel}

\subsection{System Model}
\label{System}

\begin{figure}[!t]
\centering
\includegraphics[width=85mm]{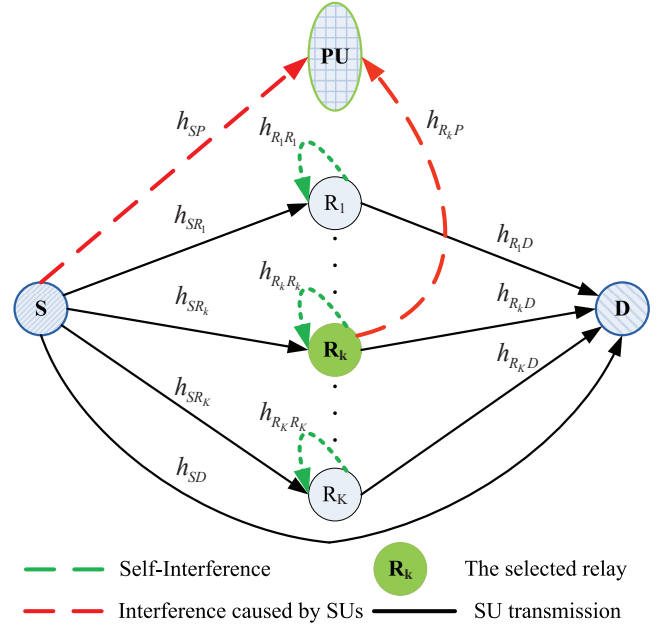}
\caption{System model of power allocation with relay selection for the cognitive full-duplex relay network.}
\label{SCRSN}
\end{figure}

We consider a cognitive relay network which consists of one SU source $S$, $K$ SU relays $R_k$ ($k = 1, \ldots, K$), one SU destination $D$, and one PU receiver $P$.
The system model for the full-duplex cognitive relay network is illustrated in Fig.~\ref{SCRSN}.
The SU relays are equipped with FD transceivers to work in the FD mode while all the other terminals work in the HD mode.
Therefore the receiver performance of each SU relay is affected by the self-interference from its transmitter since the transmit power is leaked into the received signal.

Each SU relay $R_k$ uses the AF protocol, and amplifies the received signal from $S$ with a variable gain $G_k$ and forwards the resulting signal to SU destination, $D$.
We denote $h_{SR_k}$, $h_{R_kD}$, $h_{SD}$, $h_{R_kP}$ and $h_{R_kR_k}$ by the corresponding channel coefficients of links $S \rightarrow R_k$, $R_k \rightarrow D$, $S \rightarrow D$, $R_k \rightarrow P$ and $R_k \rightarrow R_k$ which follows the independent and identically distributed (i.i.d.) Gaussian with the powers of $\sigma_{SR_k}$, $\sigma_{R_kD}$, $\sigma_{SD}$, $\sigma_{R_kP}$ and $\sigma_{R_kR_k}$.
Let  $P_S$ denote the transmit power of SU source $S$.
We also denote by $x_S(t)$, $y_{R_k}(t)$ and $y_D(t)$ the generated signal by the SU source, the transmitted signals at the SU relay and the received signals at the SU destination, respectively.

Let us consider a specific SU relay (say relay $R_k$). Fig.~\ref{SRkCRSN} illustrates the signal processing at the relay.
At time $t$, the received signals at SU relay $R_k$ and SU destination $D$ are as follows:
\beqn
y_1(t) \!=\! h_{SR_k} \sqrt{P_S} x_S(t) \!+\! h_{R_kR_k} \left(y_2(t) \!+ \! \Delta y(t)\right) \!+\! z_{R_k}(t) \label{EQN_receiveRD1}\\
y_D(t) \!=\! h_{R_kD} y_{R_k}(t) + h_{SD} \sqrt{P_S} x_S(t) + z_D (t), \label{EQN_receiveRD2} \hspace{1.2cm}
\eeqn
where $z_{R_k}(t)$ and $z_D (t)$ are the additive white Gaussian noises (AWGN) with zero mean and variances $\sigma_{R_k}^2$ and  $\sigma_D^2$, respectively;
$y_D(t)$ and $y_1(t)$ are the received signals at SU relay $R_k$ and SU destination $D$; and 
$y_2(t)$ is the received signal after the amplification.
In the following, we ignore the direct signal from the SU source to the SU destination (i.e., the second part in equation (\ref{EQN_receiveRD2})). Note that this assumption is has been used in the literature \cite{Day12b, Kim15, Zhong15, Zhong16} when there is attenuation on the direct transmission channel. 

The transmitted signals at SU relay $R_k$ is $$y_{R_k}(t) = y_2(t) + \Delta y(t),$$ where $$y_2(t) = f\left(\hat{y}_1\right) = G_k \hat{y}_1 (t-\Delta).$$
We should note that the SU relay amplifies the signal by a factor of $G_k$ and delays with duration of $\Delta$.
In the noncoherent scenario, $\Delta$ is fixed. In the coherence scenario, the delay $\Delta$ will be optimized  to minimize the interference at the PU receiver. Furthermore, $\Delta y(t)$ is the noise and follows the i.i.d. Gaussian distribution with zero mean and variance of $P_{\Delta} = \zeta P_{R_k}$ \cite{Day12, Day12b, Bliss07}. 
$G_k$ can be expressed as $$G_k = \left[P_S \left|h_{SR_k}\right|^2 + \zeta P_{R_k} \left|h_{R_kR_k}\right|^2 + \sigma_{R_k}^2\right]^{-1/2}.$$
We assume that the channel $h_{R_kR_k}$ is perfectly estimated and hence the received signal after self-interference cancellation is
\beqn
\hat{y}_1(t) \!= \sqrt{P_{R_k}} \left(y_1(t) - h_{R_kR_k} y_2(t)\right) \hspace{2.7cm} \nonumber\\
= \sqrt{P_{R_k}} \! \left[\! h_{SR_k} \! \sqrt{P_S} x_S(t) \!+\! h_{R_kR_k} \Delta y(t) \!+\! z_{R_k}(t) \!\right].
\eeqn
In the equation above, $y_2(t)$ is known at SU relay $R_k$ and therefore is used to cancel the interference.
However, the remaining $h_{R_kR_k} \Delta y(t)$ is still present at the received signals and is called the residual interference.
So we can write the transmitted signals at SU relay $R_k$ as follows:
\beqn
y_{R_k}(t) = G_k h_{SR_k} \sqrt{P_{R_k}} \sqrt{P_S} x_S(t-\Delta) + \Delta y(t) \hspace{1cm} \nonumber\\
\!+ G_k h_{R_kR_k} \!\sqrt{P_{R_k}} \Delta y(t\!-\!\Delta) \!+\! G_k \!\sqrt{P_{R_k}} z_{R_k}(t\!-\!\Delta). \label{EQN_y_Rk_trans}
\eeqn

\begin{figure}[!t]
\centering
\includegraphics[width=95mm]{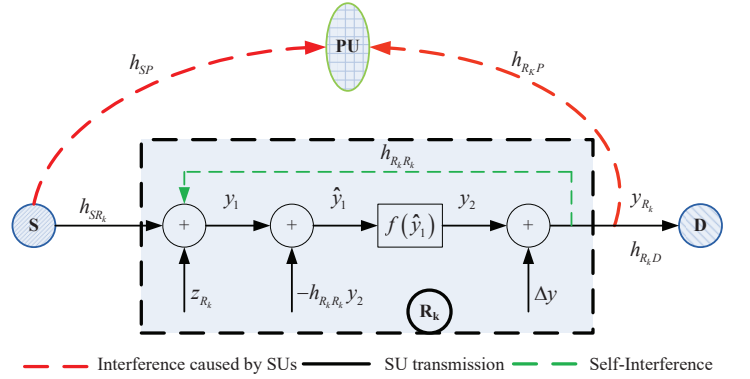} 
\caption{The process at FD relay $k$.}
\label{SRkCRSN}
\end{figure}

\section{Power Control and Relay Selection}
\label{PAFRS_Problem}

In this section, we study the problem of maximizing the rate between SU source and SU destination while protecting the PU via power control and relay selection.
Here the best relay will be selected to help the transmission from the SU source to the SU destination.

\subsection{Problem Formulation}
\label{RateOpt}

Let $\mathcal{C}_k(P_S, P_{R_k})$ denote the achieved rate of the FDCRN with relay $R_k,$ which is the function of transmit power of SU source $S$ and transmit power of SU relay $R_k$.
Assume the interference caused by the SU source and relay, $\mathcal{I}_k$ is required to be at most $\overline{\mathcal{I}}_{P}$ to protect the PU.

Now, the rate maximization problem for the selected relay $k$ can be stated as follows:

\vspace{0.05cm}
\noindent
\textbf{Problem 1:} 
\vspace{0.0cm}
\begin{equation}
\label{EQN_OPTRS}
\begin{array}{l}
 {\mathop {\max }\limits_{P_S, P_{R_k}}} \quad \mathcal{C}_k(P_S, P_{R_k})  \\ 
 \mbox{s.t.}\,\,\,\, \mathcal{I}_k\left(P_S, P_{R_k}\right) \leq \mathcal{\overline I}_{P}, 0 \leq P_S \leq P_S^{\sf max},\\
 \quad \quad 0 \leq P_{R_k} \leq P_{R_k}^{\sf max},\\
 \end{array}\!\!
\end{equation}
where $P_S^{\sf max}$ and $P_{R_k}^{\sf max}$ are the maximum power levels for the SU source and SU relay, respectively.
The first constraint on $\mathcal{I}_k\left(P_S, P_{R_k}\right)$ requires that the interference caused by the SU transmission is limited.
Moreover, the SU relay's transmit power $P_{R_k}$ must be appropriately set to achieve good tradeoff between the rate of the SU network and self-interference mitigation.

Then the relay selection is determined by
\beqn
\label{EQN_OPT_RELAY_SELEC}
k^* = {\mathop  {\arg \max }\limits_{k \in \left\{1,\ldots, K\right\}}} \quad \mathcal{C}_k^*
\eeqn
where $\mathcal{C}_k^*$ is the solution of (\ref{EQN_OPTRS}).
In the following, we show how to calculate the achieved rate, $\mathcal{C}_k(P_S, P_{R_k})$ and the interference imposed by SU transmissions, $\mathcal{I}_k\left(P_S, P_{R_k}\right)$.

\subsection{The Achievable Rate}
\label{Rate_Formu}

When SU relay $R_k$  is selected, the achievable rate of the link $S \rightarrow R_k \rightarrow D$ based on (\ref{EQN_receiveRD1}) and (\ref{EQN_receiveRD2}) is as follows:
\beqn
\mathcal{C}_k = \log_2 \left[1+ \frac{\frac{P_{R_k} \left|h_{R_kD}\right|^2}{\sigma_D^2} \frac{P_S \left|h_{SR_k}\right|^2}{\hat{\zeta} P_{R_k}+\sigma_{R_k}^2}}  {\mathcal{A}}\right]
\eeqn
where
\beqn
\mathcal{A} &=& 1+ \frac{P_{R_k} \left|h_{R_kD}\right|^2}{\sigma_D^2} + \frac{P_S \left|h_{SR_k}\right|^2}{\hat{\zeta} P_{R_k}+\sigma_{R_k}^2} \label{EQN_CAL_A}\\
\hat{\zeta} &=& \left|h_{R_kR_k}\right|^2 \zeta \label{EQN_zeta_hat}
\eeqn
Recall that we assume the direct signal from the SU source to the SU destination is negligible.

\subsection{The Imposed Interference at PU}
\label{Interference_Formu}

We now determine the interference at the PU caused by the CRN. The interference is the signals from the SU source $S$ and the selected relay $R_k:$ 
\beqn
\label{EQN_Inter}
y_I^{\sf PU} (t) \!=\! h_{SP} \sqrt{P_S} x_S(t) \!+\! h_{R_kP} y_{R_k}(t)\! +\! z_P(t)
\eeqn
where $z_P(t)$ is the AWGN with zero mean and variance $\sigma_P^2,$ and $y_{R_k}(t)$ is defined in (\ref{EQN_y_Rk_trans}).

We next derive and analyze the interference in two cases: coherent and non-coherent. In particular, we focus on coherent/non-coherent transmissions from the SU source and the SU relay to the PU receiver. All other transmissions are assumed to be non-coherent for simplicity.  In the coherent scenario, the phase information is needed for a further regulation. This information can be obtained by using methods such as the implicit feedback (using reciprocity between forward and reverse channels in a time-division-duplex system), and explicit feedback (using feedback in a frequency-division-duplex system) \cite{Mudumbai09} or the channel estimation \cite{Arslan07}.

\subsubsection{Non-coherent Scenario}

From (\ref{EQN_Inter}) and (\ref{EQN_y_Rk_trans}), the received interference at the PU caused by the SU source and the selected relay can be written as follows:
\beqn
\mathcal{I}^{\sf non}_k\left(P_S, P_{R_k}\right) = \left|h_{SP}\right|^2 P_S + \left|h_{R_kP}\right|^2 \zeta P_{R_k} \hspace{1.5cm}\nonumber\\
+ G_k^2 \left|h_{R_kP}\right|^2 \! P_{R_k} \left[ \left|h_{SR_k}\right|^2 \! P_S \!+\! \left|h_{R_kR_k}\right|^2 \! \zeta P_{R_k} \!+\! \sigma_{R_k}^2\right]  
\eeqn
After using some simple manipulations, we obtain
\beqn
\mathcal{I}^{\sf non}_k\left(P_S, P_{R_k}\right) = \left|h_{SP}\right|^2 P_S + \left|h_{R_kP}\right|^2 P_{R_k} \left(1+\zeta\right)
\eeqn

\subsubsection{Coherent Scenario}

Combining (\ref{EQN_Inter}) with (\ref{EQN_y_Rk_trans}), the received interference at the PU caused by the SU source and the selected SU relay is
\beqn
\label{EQN_I_k_coh_ori}
\mathcal{\bar{I}}^{\sf coh}_k\left(P_S, P_{R_k}, \phi\right) = \left|A+B e^{-j\phi}\right|^2
\eeqn
where 
\beqn
A = h_{SP} \sqrt{P_S} + h_{R_kP}\sqrt{\zeta P_{R_k}} = \left|A\right| \angle{\phi_A}
\eeqn
\beqn
B = \left(h_{SR_k} \sqrt{P_S} + h_{R_kR_k} \sqrt{\zeta P_{R_k}} + \frac{\sigma_{R_k}}{\sqrt{2}} (1+j)\right) \nonumber\\
\times G_k h_{R_kP} \sqrt{P_{R_k}} = \left|B\right| \angle{\phi_B}
\eeqn
and $\phi = 2 \pi f_s \Delta$, $f_s$ is the sampling frequency.

Before using $\mathcal{\bar{I}}^{\sf coh}_k\left(P_S, P_{R_k}, \phi\right)$ in the constraint of the optimization problem, we can minimize $\mathcal{\bar{I}}^{\sf coh}_k\left(P_S, P_{R_k}, \phi\right)$ over the variable $\phi$ at given $\left(P_S, P_{R_k}\right)$, i.e.,
\vspace{0.05cm}
\noindent
\vspace{0.0cm}
\begin{equation}
\label{EQN_OPT_PHI}
 {\mathop {\min }\limits_{\phi} \quad \mathcal{\bar{I}}^{\sf coh}_k\left(P_S, P_{R_k}, \phi\right) } 
\end{equation}

\vspace{0.2cm}
\noindent

\begin{theorem} \label{theorem1}
The optimal solution to (\ref{EQN_OPT_PHI}) is
\beqn
\label{EQN_I_COH}
\phi_{\sf opt} = \pi + \phi_B - \phi_A \hspace{3cm} \nonumber \\
\mathcal{I}^{\sf coh}_k\left(P_S, P_{R_k}\right) =  \mathcal{\bar{I}}^{\sf coh}_k\left(P_S, P_{R_k}, \phi_{\sf opt}\right) = \left(\left|A\right| - \left|B\right|\right)^2. 
\eeqn
\end{theorem}

\begin{proof} The proof is given in Appendix~\ref{Theo1}. \end{proof}

\section{Power Control and Relay Section in the Non-coherent Scenario}
\label{PCRS_Configuration_NonCoh}

At the SU relay, we assume the self-interference is much higher than the noise, i.e., $\hat{\zeta} P_{R_k} >> \sigma_{R_k}^2$.
Therefore, we omit the term $\sigma_{R_k}^2$ in the object function.
Moreover $\log_2(1+x)$ is a strictly increase function in $x$, so we rewrite \textbf{Problem 1} as

\vspace{0.05cm}
\noindent
\textbf{Problem 2:} 
\vspace{0.0cm}
\begin{equation}
\label{EQN_OPTRS_1}
\begin{array}{l}
 {\mathop {\max }\limits_{P_S, P_{R_k}}} \quad \mathcal{\bar C}_k(P_S, P_{R_k})  \\ 
 \mbox{s.t.}\,\,\,\, \mathcal{I}_k^{\sf non}\left(P_S, P_{R_k}\right) \leq \mathcal{\overline I}_{P},  0 \leq P_S \leq P_S^{\sf max}, \\
 \quad \quad 0 \leq P_{R_k} \leq P_{R_k}^{\sf max},\\
 \end{array}\!\!
\end{equation}
where 
\beqn
\mathcal{\bar C}_k(P_S, P_{R_k}) = \frac{\frac{P_{R_k} \left|h_{R_kD}\right|^2}{\sigma_D^2} \frac{P_S \left|h_{SR_k}\right|^2}{\hat{\zeta} P_{R_k}}}  {\mathcal{\bar A}}
\eeqn
$\mathcal{\bar A}$ is given as
\beqn\label{EQN_A_BAR}
\mathcal{\bar A} = 1+ \frac{P_{R_k} \left|h_{R_kD}\right|^2}{\sigma_D^2} + \frac{P_S \left|h_{SR_k}\right|^2}{\hat{\zeta} P_{R_k}}
\eeqn
and $\hat{\zeta}$ is calculated in (\ref{EQN_zeta_hat}).

We characterize the optimal solutions for \textbf{Problem 2} by the following lemmas.

\begin{lemma} \label{Lemma_noncoh1}
\textbf{Problem 2} is a nonconvex optimization problem for variables $\left(P_S, P_{R_k}\right)$.
\end{lemma} 

\begin{proof} The proof is provided in Appendix~\ref{Lemma01}. \end{proof}

\begin{lemma} 
Given $P_S \in \left[0, P_S^{\sf max}\right]$, \textbf{Problem 2} is a convex optimization problem in terms of $P_{R_k}$. Similarly, given $P_{R_k} \in \left[0, P_{R_k}^{\sf max}\right]$, \textbf{Problem 2} is also a convex optimization problem in terms of $P_S$. \label{lem: nonco-convex}
\end{lemma} 

\begin{proof} The proof can be found in Appendix~\ref{Lemma1}. \end{proof}

Since \textbf{Problem 2} is non-convex, we exploit alternating-optimization problem (according to Lemma \ref{lem: nonco-convex}, the problem is convex when we fix one variable and optimize the other) to solve \textbf{Problem 2}, where each step is a convex optimization problem and can be solved using standard approaches \cite{Boyd04}. Finally, we determine the best relay by solving (\ref{EQN_OPT_RELAY_SELEC}).


We now consider the special case of ideal self-interference cancellation, i.e., $\hat{\zeta} = 0$.
We characterize the optimal solutions for \textbf{Problem 1} in the special case by the following lemma. 

\begin{lemma} \label{Lemma_noncoh1_zeta0}
\textbf{Problem 1} is a convex optimization problem for variables $\left(P_S, P_{R_k}\right)$ when $\hat{\zeta} = 0$. 
\end{lemma} 

\begin{proof} The proof is provided in Appendix~\ref{Lemma01_zeta0}. \end{proof}

Based on Lemma ~\ref{Lemma_noncoh1_zeta0}, we can solve \textbf{Problem 1} when $\hat{\zeta} = 0$ by using fundamental methods \cite{Boyd04}.

\section{Power Control and Relay Selection in the Coherent Scenario}
\label{PCRS_Configuration_Coh}

Again, we assume that the self-interference is much higher than the noise at the selected relay, i.e., $\hat{\zeta} P_{R_k} >> \sigma_{R_k}^2$.
\textbf{Problem 1} can thus be reformulated as

\vspace{0.05cm}
\noindent
\textbf{Problem 3:} 
\vspace{0.0cm}
\begin{equation}
\label{EQN_OPTRS_2}
\begin{array}{l}
 {\mathop {\max }\limits_{P_S, P_{R_k}}} \quad \mathcal{\bar C}^{\sf coh}_k(P_S, P_{R_k})  \\ 
 \mbox{s.t.}\,\,\,\, \mathcal{I}_k^{\sf coh}\left(P_S, P_{R_k}\right) \leq \mathcal{\overline I}_{P}, 0 \leq P_S \leq P_S^{\sf max}, \\
 \quad \quad 0 \leq P_{R_k} \leq P_{R_k}^{\sf max},\\
 \end{array}\!\!
\end{equation}  
where 
\beqn
\mathcal{\bar C}^{\sf coh}_k(P_S, P_{R_k}) = \frac{\frac{P_{R_k} \left|h_{R_kD}\right|^2}{\sigma_D^2} \frac{P_S \left|h_{SR_k}\right|^2}{\hat{\zeta} P_{R_k}}}  {1+ \frac{P_{R_k} \left|h_{R_kD}\right|^2}{\sigma_D^2} + \frac{P_S \left|h_{SR_k}\right|^2}{\hat{\zeta} P_{R_k}}}
\eeqn
and $\hat{\zeta}$ is calculated in (\ref{EQN_zeta_hat}).

To solve \textbf{Problem 3}, the new variables are introduced as $p_S = \sqrt{P_S}$ and $p_{R_k} = \sqrt{P_{R_k}}$.
Hence \textbf{Problem 3} can be equivalently formulated as

\vspace{0.05cm}
\noindent
\textbf{Problem 4:} 
\vspace{0.0cm}
\begin{equation}
\label{EQN_OPTRS_3}
\begin{array}{l}
 {\mathop {\max }\limits_{p_S, p_{R_k}}} \quad \mathcal{\breve{C}}^{\sf coh}_k(p_S, p_{R_k})  \\ 
 \mbox{s.t.}\,\,\,\, \mathcal{I}_k^{\sf coh}\left(P_S, P_{R_k}\right) \leq \mathcal{\overline I}_{P},  0 \leq p_S \leq \sqrt{P_S^{\sf max}}, \\
 \quad \quad 0 \leq p_{R_k} \leq \sqrt{P_{R_k}^{\sf max}},\\
 \end{array}\!\!
\end{equation} 
where the objective function is written as
\beqn
\mathcal{\breve{C}}^{\sf non}_k(p_S, p_{R_k}) = \frac{\frac{p_{R_k}^2 \left|h_{R_kD}\right|^2}{\sigma_D^2} \frac{p_S^2 \left|h_{SR_k}\right|^2}{\hat{\zeta} p_{R_k}^2}}  {1+ \frac{p_{R_k}^2 \left|h_{R_kD}\right|^2}{\sigma_D^2} + \frac{p_S^2 \left|h_{SR_k}\right|^2}{\hat{\zeta} p_{R_k}^2}}
\eeqn

We give a characterization of optimal solutions for \textbf{Problem 4} by the following lemmas.

\begin{lemma} \label{Lemma_coh1} \textbf{Problem 4} is not a convex optimization problem for variable $\left(p_S,p_{R_k}\right)$. \end{lemma}

\begin{proof} The proof is in Appendix~\ref{Lemma02}. \end{proof}

\begin{lemma} Given $p_S \in \left[0, \sqrt{P_S^{\sf max}}\right]$, \textbf{Problem 4} is a convex optimization problem for variable $p_{R_k}$. Similarly, given $p_{R_k} \in \left[0, \sqrt{P_{R_k}^{\sf max}}\right]$, \textbf{Problem 4} is also a convex optimization problem for variable $p_S$. \label{lem: co-convex}\end{lemma}

\begin{proof} The proof is provided in Appendix~\ref{Lemma2}. \end{proof}

Based on Lemma \ref{lem: co-convex}, we again develop the alternating-optimization strategy to solve \textbf{Problem 4}, where each step is a convex optimization problem and can be solved using basic approaches \cite{Boyd04}. The relay selection is then determined by solving (\ref{EQN_OPT_RELAY_SELEC}).


We now investigate the special case of ideal self-interference cancellation, i.e., $\hat{\zeta} = 0$.
We then characterize the optimal solutions for \textbf{Problem 4} in the special case by the following lemma.

\begin{lemma} \label{Lemma_coh2_zeta0}
\textbf{Problem 4} is a convex optimization problem for variables $\left(P_S, P_{R_k}\right)$ when $\hat{\zeta} = 0$.
\end{lemma} 

\begin{proof} The proof can be found in Appendix~\ref{Lemma02_zeta0}. \end{proof}

According to Lemma ~\ref{Lemma_coh2_zeta0}, we can solve \textbf{Problem 4} in this special case by using standard approaches \cite{Boyd04}.

\vspace{10pt}
\section{Numerical Results}
\label{Results}

In the numerical evaluation, we set the key parameters for the FDCRN as follows.
We assume that each link is a Rayleigh fading channel with variance one (i.e., $\sigma_{SR_k}$ = $\sigma_{R_kD}$ = 1), except the negligible $\sigma_{SD}$ = 0.1.
The noise power at every node is also set to be one.
The channel gains for the links of the SU relay-PU receiver and SU source-PU receiver are assumed to be Rayleigh-distributed with  variances $\left\{\sigma_{SP}, \sigma_{R_kP}\right\} \in \left[0.8, 1\right]$.
We also assume that the impact of imperfect channel estimation is included in only one parameter, i.e., $\zeta$.
Due to the space constraint, we only report some essential results, more detailed results can be found in the online technical report \cite{TanTechreport}.

\begin{table*} 
\centering
\caption{Achievable rate vs $\mathcal{\bar I}_P$ ($P_{\sf max} = 20 dB$, $\zeta = 0.001$)}
\label{table1}
\begin{tabular}{|c|c|c|c|c|c|c|c|}
\hline 
\multicolumn{2}{|c|}{$\mathcal{\bar I}_P$ (dB)}     
       & 0   & 2   & 4   & 6   & 8 & 10\tabularnewline
\hline 
\hline 
$\zeta = 0.001$,      & Optimal & 4.3646  &  5.1933  &  5.5533  &  5.6944  &  5.8162  &  5.9155 \tabularnewline
\cline{2-8} 
Coherent  & Greedy  & 4.3513  &  5.1807  &  5.5496  &  5.6811  &  5.8131  &  5.8826  \tabularnewline
\cline{2-8} 
scenario & $\Delta \mathcal{C} (\%)$ & 0.3047   & 0.2426   & 0.0666   & 0.2336   & 0.0533  &  0.5562  \tabularnewline
\hline 
\hline 
$\zeta = 0.001$,    & Optimal & 1.2390  &  1.6946  &  2.2118  &  2.7753  &  3.3718  &  3.9902  \tabularnewline
\cline{2-8} 
Non-coherent  & Greedy  & 1.2309  &  1.6856  &  2.2018  &  2.7650  &  3.3610  &  3.9791  \tabularnewline
\cline{2-8} 
scenario & $\Delta \mathcal{C} (\%)$ & 0.6538  &  0.5311  &  0.4521  &  0.3711  &  0.3203  &  0.2782  \tabularnewline
\hline
\end{tabular}
\end{table*}

We first demonstrate the efficacy of the proposed algorithms by comparing their achievable rate performances with those obtained by the optimal brute-force search algorithms.
Numerical results are presented for both coherent and non-coherent scenarios where we set $P_S^{\sf max} = P_{R_k}^{\sf max} = P_{\sf max}$ for simplicity.
In Table~\ref{table1}, we consider the scenario with $\zeta$ = 0.001, 8 SU relays and $P_{\sf max} = 20$ dB.
We compare the achievable rate of the proposed and optimal algorithms for $\mathcal{\bar I}_P = \left\{0, 2, 4, 6, 8, 10\right\}$ dB.
These results confirm that our proposed algorithms achieve rate very close to that attained by the optimal solutions for both coherent and non-coherent scenarios (i.e., the errors are lower than 1\%).

We then consider a FDCRN $8$ SU relays with $\zeta$ = 0, 0.001, 0.01, and 0.4, which represent ideal, high, medium and low Quality of Self-Interference Cancellation (QSIC), respectively. The tradeoffs between the achievable rate of the FDCRN and the interference constraint are shown in Figs.~\ref{1_Rate_vs_I_bar_P_P_max_10152025_zeta_0}, \ref{1_Rate_vs_I_bar_P_P_max_10152025_zeta_0001}, \ref{1_Rate_vs_I_bar_P_P_max_10152025_zeta_001} and \ref{2_Rate_vs_I_bar_P_P_max_10152025_zeta_04} under different values of $\zeta.$ In these numerical results,  we chose $P_S^{\sf max} = P_{R_k}^{\sf max} = P_{\sf max}$ for simplicity.

We have the following observations from these numerical results. 
\begin{itemize}

\item The achievable rates of the coherent mechanism are always significantly higher than those of the  non-coherent mechanism.
This is because the phase is carefully regulated to reduce the interference at the PU receiver imposed by the SU transmissions, which allows higher transmit power both at the SU source and the SU relay.

\item As expected, the achievable rate decreases as the QSIC increases due to the increase of self-interference at the FD relay.

\item In all cases, if we increase $\mathcal{\bar I}_P$, the performance in terms of data rate increases. 
Because the feasible range of $\left\{P_S, P_{R_k}\right\}$ is much larger. 
However there is the threshold for $\mathcal{\bar I}_P$ where we cannot obtain the higher data rate when we increase $\mathcal{\bar I}_P$ (see Fig.~\ref{2_Rate_vs_I_bar_P_P_max_10152025_zeta_04}).
Because to obtain the higher performance of data rate with higher $\mathcal{\bar I}_P$, we shall increase $P_S$ and $P_{R_k}$ .
However the self-interference is also higher due to the increase of $P_{R_k}$ which results to decrease the data rate performance.

\end{itemize}

\begin{figure}[!t]
\centering
\includegraphics[width=70mm]{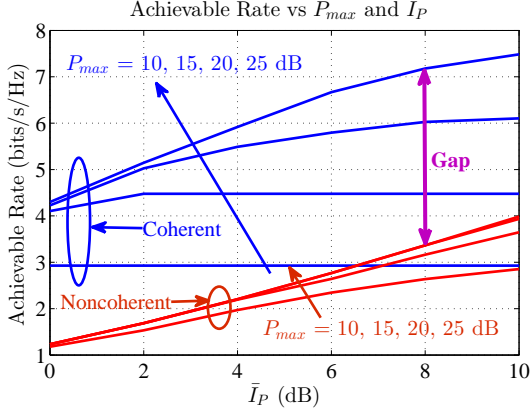} 
\caption{Achievable rate versus the interference constraint $\mathcal{\bar I}_P$
 for $K = 8$, $\zeta = 0$, $P_{\sf max} = \left\{10, 15, 20, 25\right\}$ dB, and both coherent and non-coherent scenarios.}
\label{1_Rate_vs_I_bar_P_P_max_10152025_zeta_0}
\end{figure}

\begin{figure}[!t]
\centering
\includegraphics[width=70mm]{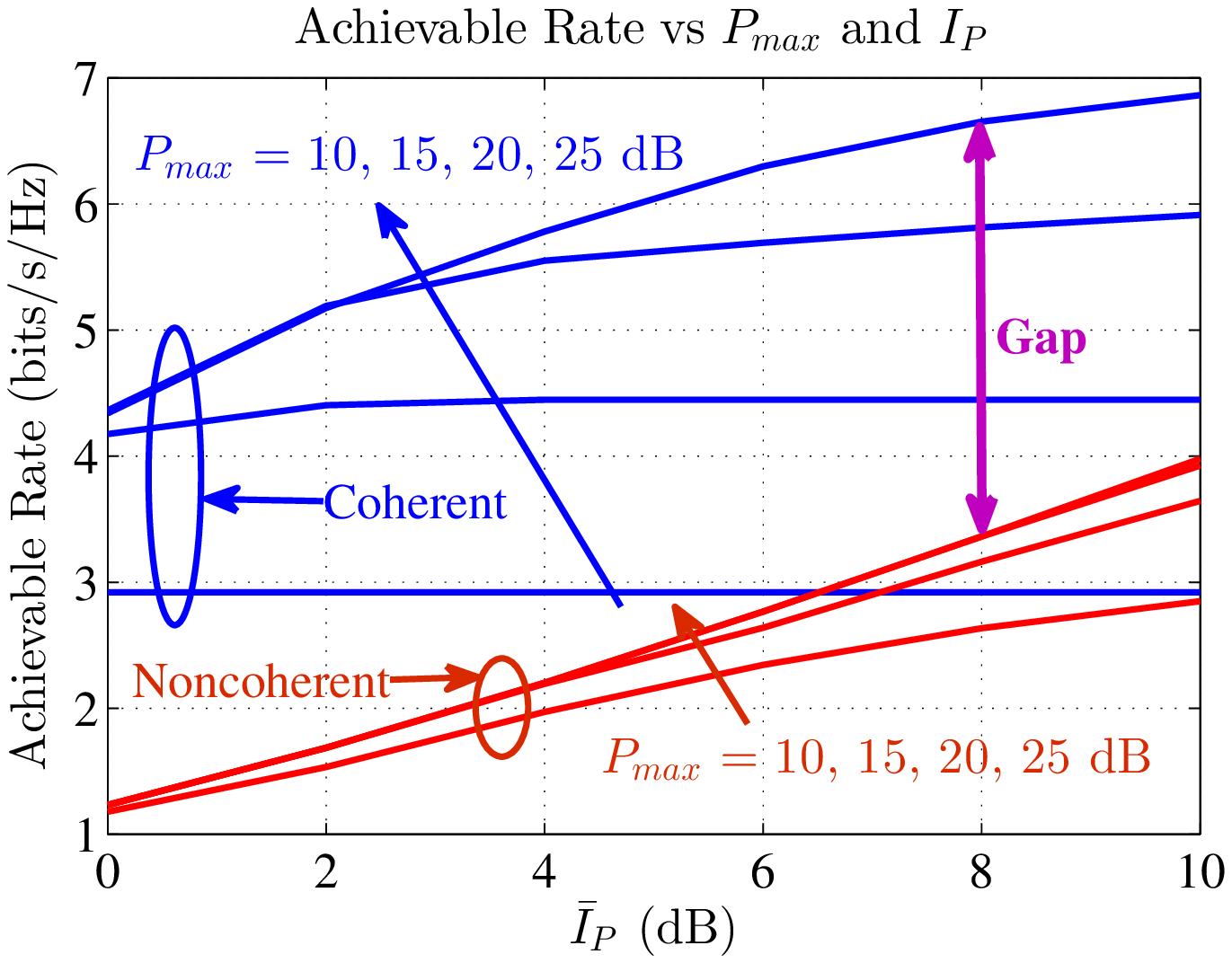}
\caption{Achievable rate versus the interference constraint $\mathcal{\bar I}_P$
 for $K = 8$, $\zeta = 0.001$, $P_{\sf max} = \left\{10, 15, 20, 25\right\}$ dB, and both coherent and non-coherent scenarios.}
\label{1_Rate_vs_I_bar_P_P_max_10152025_zeta_0001}
\end{figure}

\begin{figure}[!t]
\centering
\includegraphics[width=70mm]{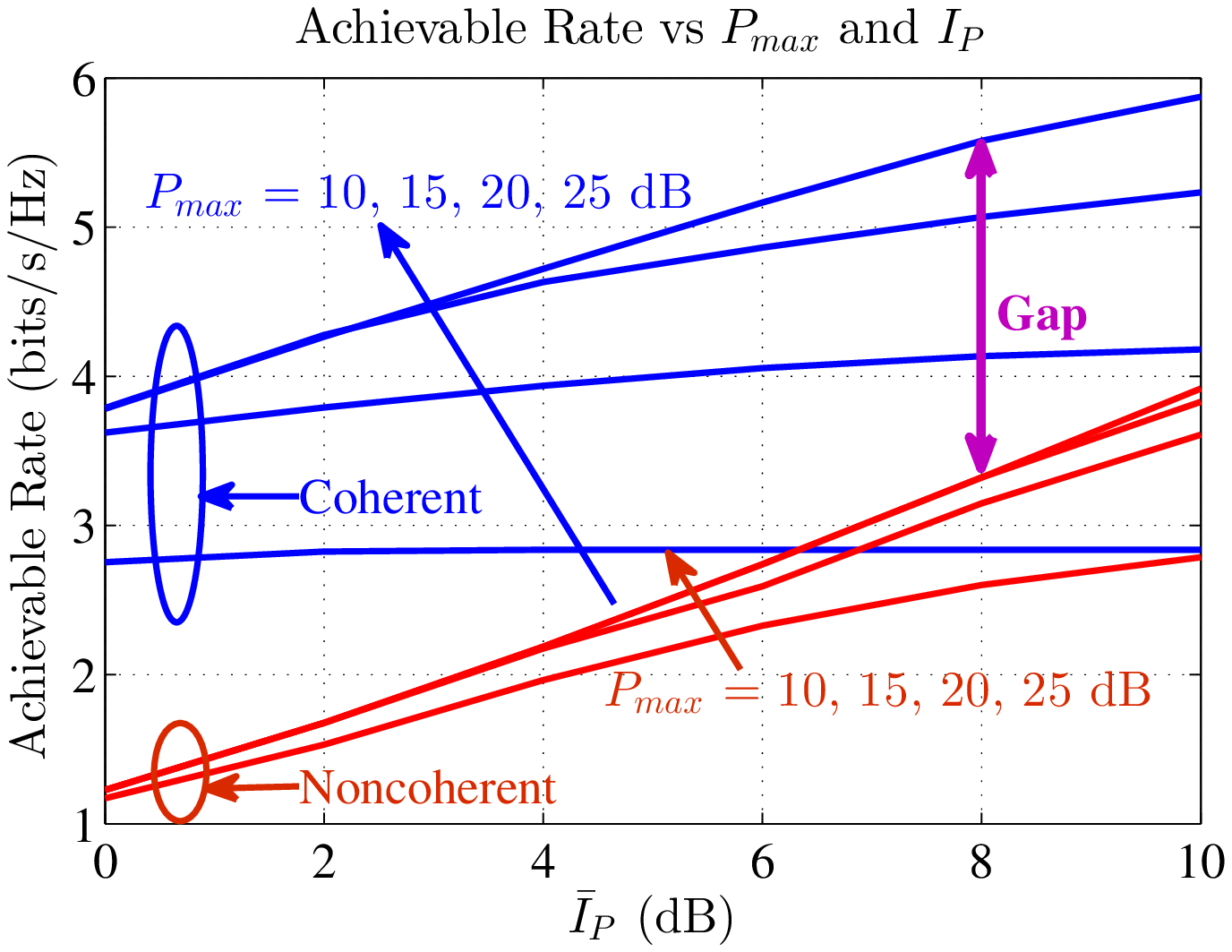}
\caption{Achievable rate versus the interference constraint $\mathcal{\bar I}_P$
 for $K = 8$, $\zeta = 0.01$, $P_{\sf max} = \left\{10, 15, 20, 25\right\}$ dB, and both coherent and non-coherent scenarios.}
\label{1_Rate_vs_I_bar_P_P_max_10152025_zeta_001}
\end{figure}

\begin{figure}[!t]
\centering
\includegraphics[width=70mm]{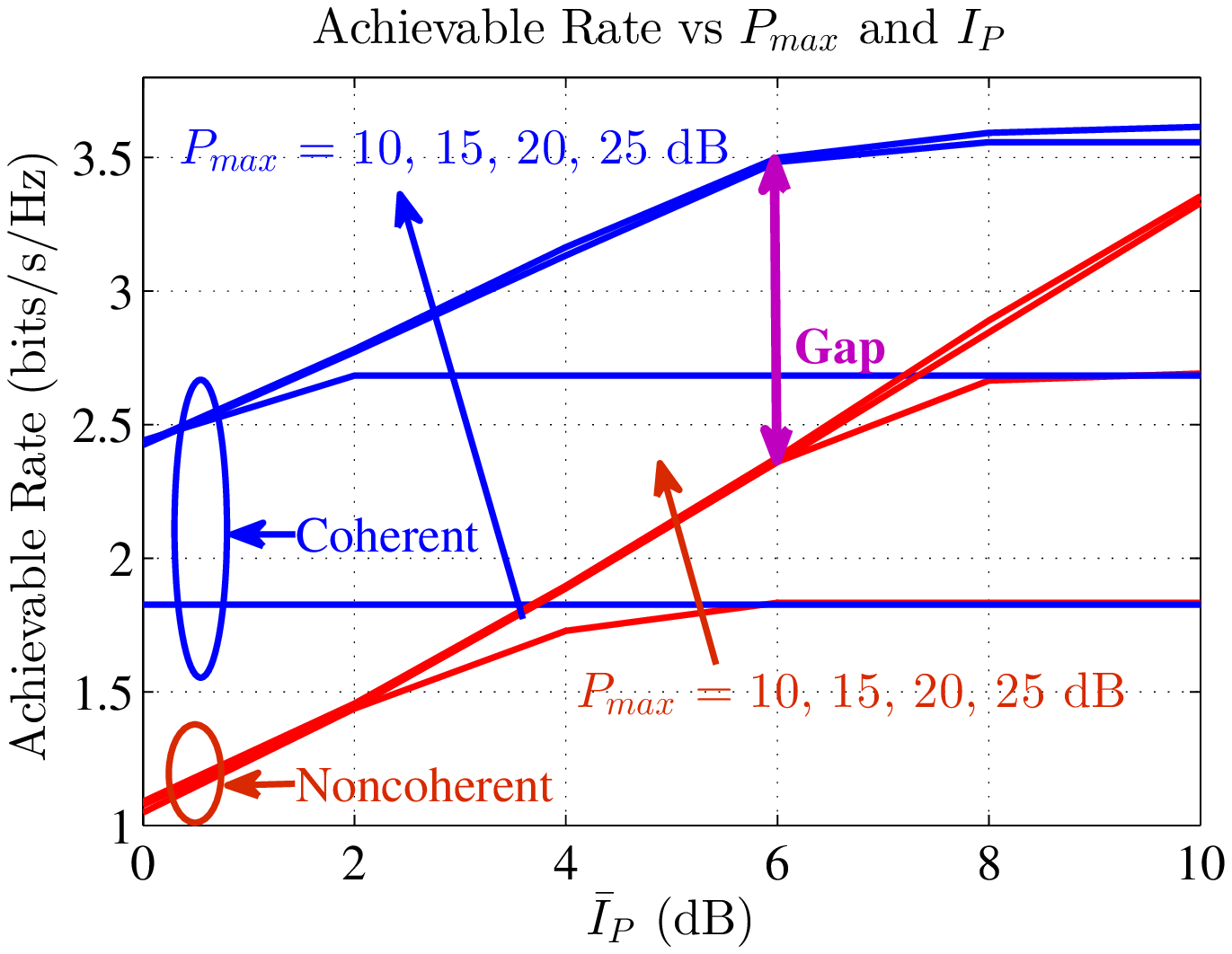}
\caption{Achievable rate versus the interference constraint $\mathcal{\bar I}_P$
 for $K = 8$, $\zeta = 0.4$, $P_{\sf max} = \left\{10, 15, 20, 25\right\}$ dB, and both coherent and non-coherent scenarios.}
\label{2_Rate_vs_I_bar_P_P_max_10152025_zeta_04}
\end{figure}

\begin{figure}[!t]
\centering
\includegraphics[width=70mm]{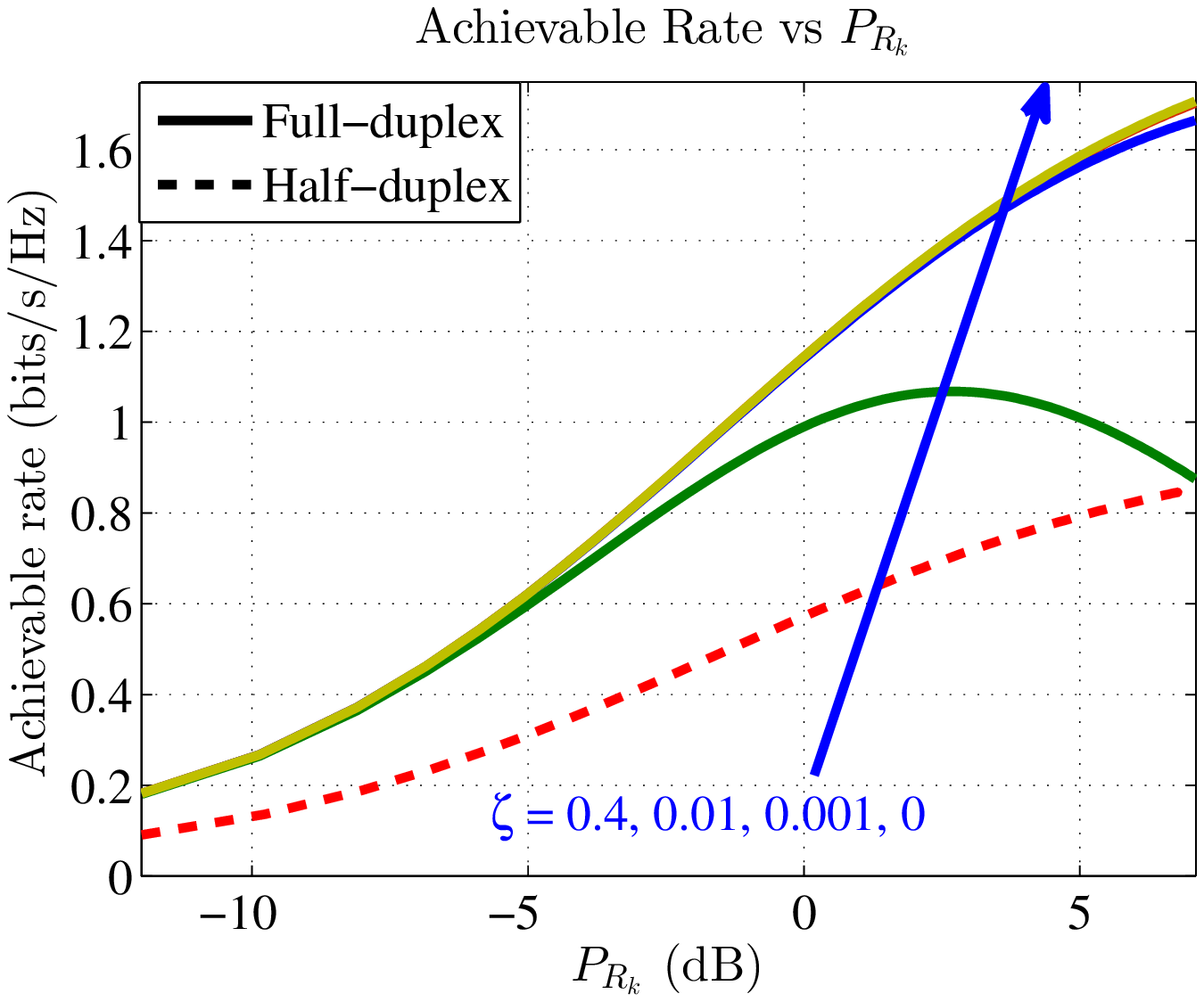}
\caption{Achievable rate versus the transmitted powers of SU relay $P_{R_k}$
for fixed $P_S = 5$ dB, $K = 10$, $\mathcal{\bar I}_P = 8$ dB, $P_{\sf max} = 25$ dB, and the non-coherent scenario.}
\label{Rate_vs_P_Rk_Non_P_S_5dB}
\end{figure}

\begin{figure}[!t]
\centering
\includegraphics[width=70mm]{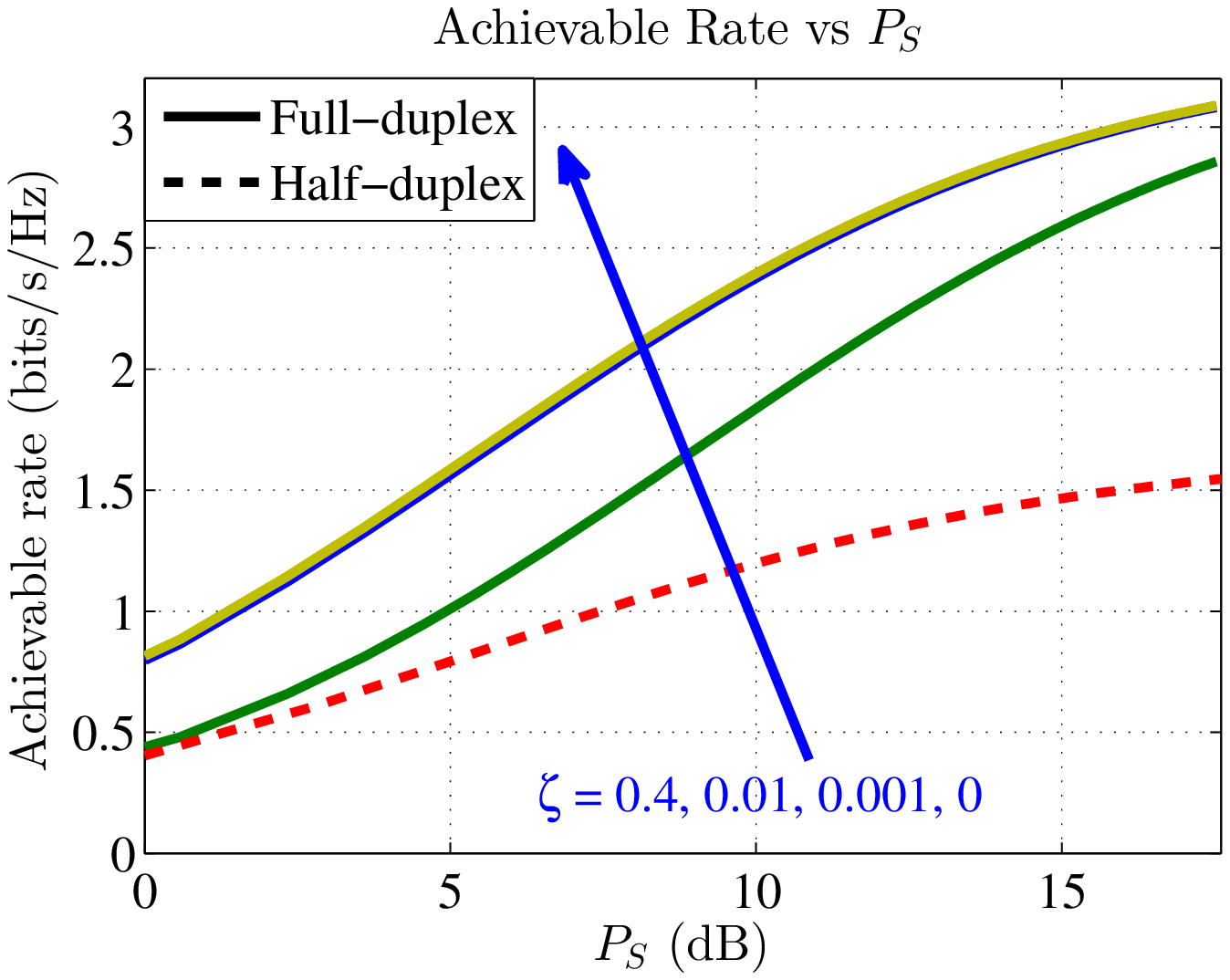}
\caption{Achievable rate versus the transmitted powers of SU source $P_S$
for fixed $P_{R_k} = 5$ dB, $K = 10$, $\mathcal{\bar I}_P = 8$ dB, $P_{\sf max} = 25$ dB, and the non-coherent scenario.}
\label{Rate_vs_P_S_Non_P_Rk_5dB}
\end{figure}

\begin{figure}[!t]
\centering
\includegraphics[width=70mm]{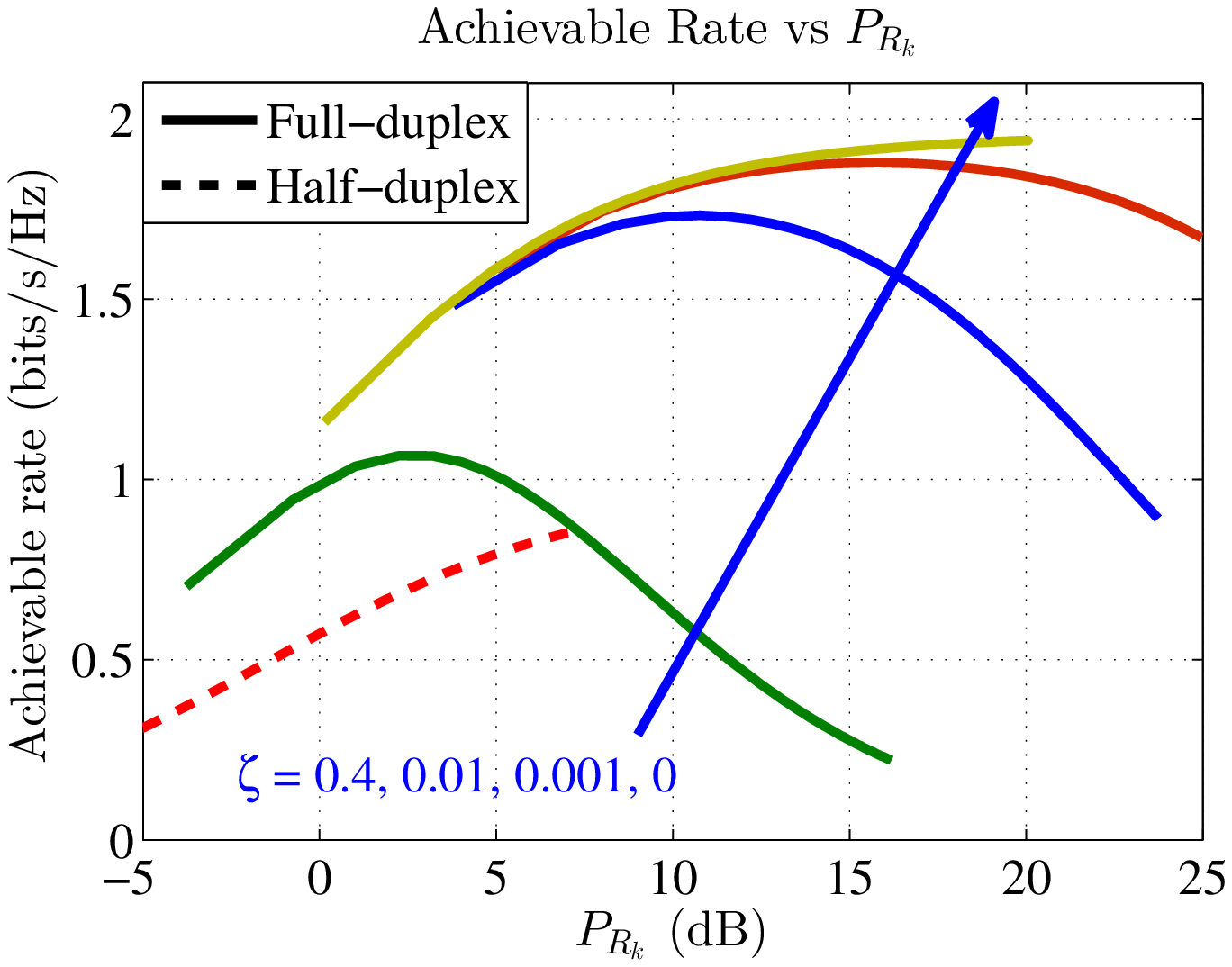}
\caption{Achievable rate versus the transmitted powers of SU relay $P_{R_k}$
for fixed $P_S = 5$ dB, $K = 10$, $\mathcal{\bar I}_P =8$ dB, $P_{\sf max} = 25$ dB, and the coherent scenario.}
\label{Rate_vs_P_Rk_Coh_P_S_5dB}
\end{figure}

\begin{figure}[!t]
\centering
\includegraphics[width=70mm]{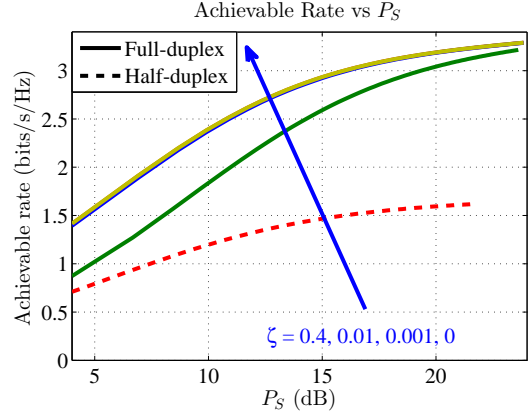}
\caption{Achievable rate versus the transmitted powers of SU source $P_S$
for fixed $P_{R_k} = 5$ dB, $K = 10$, $\mathcal{\bar I}_P =8$ dB, $P_{\sf max} = 25$ dB, and the coherent scenario.}
\label{Rate_vs_P_S_Coh_P_Rk_5dB}
\end{figure}

We now show the achievable rates of the FDCRN under different values of $P_{R_k}$ when fixing $P_S = 5$ dB in Fig.~\ref{Rate_vs_P_Rk_Non_P_S_5dB}.
The channel gains of the links of the SU relay-PU receiver and SU source-PU receiver were assumed to be Rayleigh-distributed with  variances $\left\{\sigma_{R_kP}, \sigma_{SP}\right\} \in \left[0.8, 1\right]$.
Fig.~\ref{Rate_vs_P_Rk_Non_P_S_5dB} evaluates the non-coherent scenario, $K = 10$ SU relays and $\mathcal{\bar I}_P =8$ dB.
For each value of $\zeta$ ($\zeta = \left\{0.4, 0.01, 0.001, 0\right\}$), there exists the corresponding optimal SU relay transmit power $P_{R_k}^*$ ($P_{R_k}^* = \left\{2.6995, 7.0891, 7.0891, 7.0891\right\}$ dB) where the SUs achieves the largest rate.
Furthermore, the achieved rate significantly decreases when the $P_{R_k}$ deviates from the optimal value $P_{R_k}^*$.
It is also easily observed that all four cases (low, medium, high and ideal QSIC) have similar behaviors and achieve higher rate than the half-duplex case. For low QSIC (i.e., $\zeta = 0.4$),  the rate first increases then decreases as $P_{R_k}$ increases where the rate decrease is due to the strong self-interference.

Fig.~\ref{Rate_vs_P_S_Non_P_Rk_5dB} illustrates the achievable rates of the cognitive system against $P_S$ for a fixed $P_{R_k} = 5$ dB in the non-coherent scenario.
Here we also consider the cognitive radio network with the non-coherent scenario, $K = 10$ SU relays and the parameter setting of $\mathcal{\bar I}_P = 8$ dB.
Moreover, we compare the achievable rate of our proposed power allocation for the FDCRN where SU relays can perform simultaneously reception and transmission and the half-duplex cognitive radios which uses different time slots for transmission and reception.
The results from both Figs.~\ref{Rate_vs_P_Rk_Non_P_S_5dB} and \ref{Rate_vs_P_S_Non_P_Rk_5dB} confirm that the proposed power allocation for the FDCRN outperforms the HDCRN at the corresponding optimal value of $P_S^*$ (or $P_{R_k}^*$) required by our proposed scheme.

We consider the cognitive radio network with the coherent scenario, $K = 10$ SU relays and the parameter setting of $\mathcal{\bar I}_P = 8$ dB.
Fig.~\ref{Rate_vs_P_Rk_Coh_P_S_5dB} demonstrates the achievable rates of the proposed FDCRN under different values of $P_{R_k}$ when fixing $P_S = 5$ dB; while Fig.~\ref{Rate_vs_P_S_Coh_P_Rk_5dB} demonstrates the achievable rates of the proposed FDCRN under different values of $P_S$ when fixing $P_{R_k} = 5$ dB.
In the coherent scenario, we also have the same observations as those in the non-coherent scenario.
However, the feasible range of $P_{R_k}$ (or $P_S$) in the coherent scenario is much larger than that in the non-coherent scenario.
That is the result of coherent phase regulation which then improves the achievable rates.
Recall that more numerical results can be found in the online technical report \cite{TanTechreport}.

\section{Conclusion}
\label{conclusion} 

This paper studied power control and relay selection in FDCRNs. We formulated the rate maximization problem, analyzed the achievable rate under the interference constraint, and proposed joint power control and relay selection algorithms based on alternative optimization.
The design and analysis have taken into account the self-interference of the FD transceiver, and included the both coherent and non-coherent scenarios. Numerical results have been presented to demonstrate the impacts
of the levels of self-interference and the significant gains of the coherent mechanism. 
Moreover, we have shown that the proposed FDCRN achieves significantly higher data rate than the conventional HD schemes, which confirms that the FDCRN can efficiently exploit the FD communication capability. 


\section*{Acknowledgment}
This work was supported in part by the NSF under Grant CNS-1262329, ECCS-1547294, ECCS-1609202, and the U.S. Office of Naval Research (ONR) under Grant N00014-15-1-2169.

\appendices

\section{Proof of Theorem~ \ref{theorem1}}
\label{Theo1}

We start the proof by defining the following quantities.
Let us define $\tilde{B} = B e^{-j\phi} = \left|B\right| \angle \phi_{\tilde{B}}$ where $\left|\tilde{B}\right| = \left|B\right|$ and $\phi_{\tilde{B}} = \phi_B - \phi$.
Let the subscripts $R$ and $I$ (at $a^R$ and $a^I$ ) be the real and image of complex $a$.
Note that we can obtain the following equations
\beqn
\left|A\right|^2 &=& \left(A^R\right)^2 + \left(A^I\right)^2 \label{EQN_A_abs}\\
\left|\tilde{B}\right|^2 &=& \left(\tilde{B}^R\right)^2 + \left(\tilde{B}^I\right)^2 \label{EQN_B_abs}\\
A^R &=& \left|A\right| {\sf Cos} \left(\phi_A\right) \label{EQN_A_Real}\\
A^I &=& \left|A\right| {\sf Sin}\left(\phi_A\right) \label{EQN_A_Image}\\
\tilde{B}^R &=& \left|\tilde{B}\right| {\sf Cos}\left(\phi_{\tilde{B}}\right) =  \left|B\right| {\sf Cos}\left(\phi_{\tilde{B}}\right) \label{EQN_B_Real}\\
\tilde{B}^I &=& \left|\tilde{B}\right| {\sf Sin}\left(\phi_{\tilde{B}}\right) = \left|B\right| {\sf Sin}\left(\phi_{\tilde{B}}\right) \label{EQN_B_Image}
\eeqn

From (\ref{EQN_I_k_coh_ori}), we can rewrite the $\mathcal{\bar{I}}^{\sf coh}_k\left(P_S, P_{R_k}, \phi\right)$ as follows:
\beqn
\mathcal{\bar{I}}^{\sf coh}_k \left(P_S, P_{R_k}, \phi\right) &=& \left|A^R + j A^I + \tilde{B}^R + j\tilde{B}^I\right|^2 \\
&=& \left|A^R + \tilde{B}^R + j \left(A^I + \tilde{B}^I\right)\right|^2 \\
&=& \left[\!\left(\!A^R\!\right)^2 \!+\! \left(\!\tilde{B}^R\!\right)^2 \!+\! \left(\!A^I\!\right)^2 \! + \!\left(\!\tilde{B}^I\!\right)^2 \right.\\
&& \left.+ 2 \left(A^R \tilde{B}^R + A^I \tilde{B}^I\right)\right] 
\eeqn
We substitute the parameters from (\ref{EQN_A_abs})--(\ref{EQN_B_Image}) to the above result.
After using some simple manipulations, we get
\beqn
\label{EQN_I_bar_final}
\mathcal{\bar{I}}^{\sf coh}_k \left(P_S, P_{R_k}, \phi\right) \!=\! \left[\!\left|A\right|^2 \!+\! \left|B\right|^2 \!+\! 2 \left|A\right| \left|B\right|{\sf Cos} \left(\phi_A \!-\! \phi_{\tilde{B}}\right)\!\right] 
\eeqn

From (\ref{EQN_I_bar_final}), we can obtain the minimum $\mathcal{\bar{I}}^{\sf coh}_k \left(P_S, P_{R_k}, \phi\right)$ when ${\sf Cos} \left(\phi_A -\phi_{\tilde{B}}\right) = -1$.
Therefore, $\phi_A -\phi_{\tilde{B}} = \phi_A -\left(\phi_B - \phi\right)= \pi$.
Finally, we get
\beqn
\phi_{\sf opt} &=& \pi + \phi_B - \phi_A\\
\mathcal{\bar{I}}^{\sf coh}_k\left(P_S, P_{R_k}, \phi_{\sf opt}\right) &=& \left(\left|A\right| - \left|B\right|\right)^2
\eeqn
We now complete the proof of Theorem ~\ref{theorem1}.

\section{Proof of Lemma ~\ref{Lemma_noncoh1}}
\label{Lemma01}

In the following, we use contradiction to prove that \textbf{Problem 2} is not the strictly convex optimization problem.
We can easily confirm that the constraints in \textbf{Problem 2} are convex sets due to the fact that all of them are linear functions. 
To prove \textbf{Problem 2} is convex optimization problem, we must prove that the function $\mathcal{\bar C}_k(P_S, P_{R_k})$ is concave.
We rewrite $\mathcal{\bar C}_k(P_S, P_{R_k})$ as follows:
\beqn\label{EQN_C_bar}
\mathcal{\bar C}_k(P_S, P_{R_k}) = \frac{\frac{\left|h_{R_kD}\right|^2}{\sigma_D^2} \frac{\left|h_{SR_k}\right|^2}{\hat{\zeta}}}  {f\left(P_S, P_{R_k}\right)}
\eeqn
where
\beqn \label{EQN_g_FUNC}
f\left(P_S, P_{R_k}\right) = \frac{1}{P_S}+ \frac{P_{R_k} \left|h_{R_kD}\right|^2}{P_S \sigma_D^2} + \frac{\left|h_{SR_k}\right|^2}{\hat{\zeta} P_{R_k}}
\eeqn

We can see that if both the terms in (\ref{EQN_C_bar}) are concave, then $\mathcal{\bar C}_k(P_S, P_{R_k})$ is concave.
Here the first term, $\frac{P_S \left|h_{SD}\right|^2}{\sigma_D^2}$ is concave.
To prove that $\mathcal{\bar C}_k(P_S, P_{R_k})$ is concave, we can instead prove that $\overline{f}\left(P_S, P_{R_k}\right) = 1/f\left(P_S, P_{R_k}\right)$ is a concave function.
To do so, we first determine the Hessian matrix of $\overline{f}\left(P_S, P_{R_k}\right)$ which can be expressed as
\beqn
\mathcal{H}_f = \left[\begin{array}{cc} 
\mathcal{H}_{11} & \mathcal{H}_{12} \\
\mathcal{H}_{21} & \mathcal{H}_{22}
\end{array}\right]
\eeqn

Here $\left\{\mathcal{H}_{ij}\right\}$, $i, j \in \left\{1,2\right\}$ are given as
\beqn
\label{EQN_H_11_PROOF}
\begin{array}{llll} 
\mathcal{H}_{11} & = & \frac{\partial^2 \overline{f}}{\partial P_S^2} & = -\frac{f \frac{\partial^2 f}{\partial P_S^2} - 2 \left(\frac{\partial f}{\partial P_S}\right)^2}{f^3}  \label{EQN_H_11_PROOF}\\
\mathcal{H}_{12} & = & \frac{\partial^2 \overline{f}}{\partial P_S \partial P_{R_k}}  & = -\frac{f \frac{\partial^2 f}{\partial P_S \partial P_{R_k}} - 2 \frac{\partial f}{\partial P_S} \frac{\partial f}{\partial P_{R_k}}}{f^3} \label{EQN_H_12_PROOF} \\
\mathcal{H}_{21} & = & \frac{\partial^2 \overline{f}}{\partial P_{R_k} \partial P_S}  & = -\frac{f \frac{\partial^2 f}{\partial P_{R_k} \partial P_S} - 2 \frac{\partial f}{\partial P_S} \frac{\partial f}{\partial P_{R_k}}}{f^3} \label{EQN_H_21_PROOF} \\
\mathcal{H}_{22} & = & \frac{\partial^2 \overline{f}}{\partial P_{R_k}^2}  & = -\frac{f \frac{\partial^2 f}{\partial P_{R_k}^2} - 2 \left(\frac{\partial f}{\partial P_{R_k}}\right)^2}{f^3}  \label{EQN_H_22_PROOF}
\end{array}
\eeqn
where 
\beqn
\label{EQN_f_derivatives}
\begin{array}{llll} 
\frac{\partial^2 f}{\partial P_S^2} & = &\frac{2}{P_S^3} + \frac{2 P_{R_k} \left|h_{R_k D}\right|^2}{\sigma_D^2 P_S^3}\\
\frac{\partial^2 f}{\partial P_S \partial P_{R_k}} & = &- \frac{\left|h_{R_k D}\right|^2}{\sigma_D^2 P_S^2}\\
\frac{\partial^2 f}{\partial P_{R_k} \partial P_S} & = &- \frac{\left|h_{R_k D}\right|^2}{\sigma_D^2 P_S^2} \\
\frac{\partial^2 f}{\partial P_{R_k}^2} & = &\frac{2  \left|h_{S R_k}\right|^2}{\hat{\zeta} P_{R_k}^3} \\
\frac{\partial f}{\partial P_{R_k}} & = &\frac{\left|h_{R_kD}\right|^2}{P_S \sigma_D^2} - \frac{\left|h_{SR_k}\right|^2}{\hat{\zeta} P_{R_k}^2} \\
\frac{\partial f}{\partial P_S} & = &\frac{-1}{P_S^2} - \frac{\left|h_{R_kD}\right|^2 P_{R_k}}{P_S^2 \sigma_D^2} 
\end{array}
\eeqn
We can easily observe that $\mathcal{H}_f$ is the symmetric matrix, i.e., $\mathcal{H}_{12} = \mathcal{H}_{21}$.
Because $\frac{\partial^2 f}{\partial P_S \partial P_{R_k}} = \frac{\partial^2 f}{\partial P_{R_k} \partial P_S} = - \frac{\left|h_{R_k D}\right|^2}{\sigma_D^2 P_S^2}$.

According to the Sylvester's criterion \cite{Horn12}, the Hessian matrix $\mathcal{H}_f$ is negative definite iff $\mathcal{H}_{11} < 0$ and $\mathcal{H}_{11} \mathcal{H}_{22} - \mathcal{H}_{12} \mathcal{H}_{21} > 0$. 
However we can choose the values of $\left(P_S, P_{R_k}\right)$ ($P_S \in \left[0, \tilde{P}_S\right)$) such that $\mathcal{H}_{11} \mathcal{H}_{22} - \mathcal{H}_{12} \mathcal{H}_{21} < 0$.
Here the proof and the quantity of $\tilde{P}_S$ are given in Appendix~\ref{Lemma01_1}.
In this case, the Hessian matrix $\mathcal{H}_f$ is indefinite.
Hence the function $\overline{f}\left(P_S, P_{R_k}\right)$ is not the strictly convex function.
Thus \textbf{Problem 2} is not the strictly convex optimization problem.
So Lemma ~\ref{Lemma_noncoh1} is completely done.

\section{Proof of $\mathcal{H}_{11} \mathcal{H}_{22} - \mathcal{H}_{12} \mathcal{H}_{21} < 0$ when $P_S \in \left[0, \tilde{P}_S\right)$}
\label{Lemma01_1}

In this section, we prove that for any given $P_{R_k}$ we can choose the value of $P_S$ ($P_S \in \left[0, \tilde{P}_S\right)$) such that 
$\mathcal{H}_{11} \mathcal{H}_{22} - \mathcal{H}_{12} \mathcal{H}_{21} < 0$, where $\tilde{P}_S = \tilde{P}_{S2}$ is from (\ref{EQN_P_S_2_PROOF}).
We define $\mathcal{SC}_1$ as follows:
\beqn
\label{EQN_SC_1}
\mathcal{SC}_1 = \mathcal{H}_{11} \mathcal{H}_{22} - \mathcal{H}_{12} \mathcal{H}_{21} = \mathcal{H}_{11} \mathcal{H}_{22} - \left(\mathcal{H}_{12}\right)^2
\eeqn

We substitute all parameters in (\ref{EQN_f_derivatives}) to $\mathcal{H}_{11}$, $\mathcal{H}_{12}$,  and $\mathcal{H}_{22}$ at (\ref{EQN_H_11_PROOF}).
After using some manipulations, we obtain
\beqn
\mathcal{H}_{11}\!\! &\!\! = \!\!& \!\!\frac{-2 \left|h_{S R_k}\right|^2}{f^3 \hat{\zeta} P_{R_k} P_S^3 \sigma_D^2} \left(\sigma_D^2 + \left|h_{R_k D}\right|^2 P_{R_k}\right) \label{EQN_H111222_1}\\
\mathcal{H}_{22}\!\!  & \!\!= \!\!& \!\! \frac{-2}{f^3 P_S^2 \sigma_D^2} \left[- \frac{\left|h_{R_k D}\right|^4}{\sigma_D^2} \right. \nonumber\\
& \!\!+\!\! & \left.\frac{\left|h_{S R_k}\right|^2  P_S \left(\sigma_D^2 + 3\left|h_{R_k D}\right|^2 P_{R_k}\right)}{\hat{\zeta} P_{R_k}^3}   \!\right] \label{EQN_H111222_2}\\
\mathcal{H}_{12}\!\!  & \!\!= \!\!& \!\!\frac{-1}{f^3 P_S^2 \sigma_D^2} \left[ - \frac{\left|h_{R_k D}\right|^2 \left|h_{SR_k}\right|^2}{\hat{\zeta} P_{R_k}^2}\right.\nonumber\\
& \!\!+\!\! &\!\!\left.\!\left(\!\frac{\left|h_{R_k D}\right|^2}{P_S} \!-\!\frac{2\left|h_{SR_k}\right|^2\! \sigma_D^2}{\hat{\zeta} P_{R_k}^2}\!\right) \!\!\left(\!1\!+\!\frac{\left|h_{R_k D}\right|^2 \!P_{R_k}}{\sigma_D^2}\!\right)\!\right] \label{EQN_H111222_3}
\eeqn

From (\ref{EQN_H111222_1}), (\ref{EQN_H111222_2}), (\ref{EQN_H111222_3}) and (\ref{EQN_SC_1}), we can get
\beqn
\label{EQN_SC_1}
\mathcal{SC}_1 = \frac{\left|h_{R_k D}\right|^2}{P_S^6 \sigma_D^2}  \left(a P_S^2 + b P_S + c\right)
\eeqn
where 
\beqn
a &=& \frac{\left|h_{S R_k}\right|^4}{\hat{\zeta}^2 P_{R_k}^3} \left(12 + \frac{11\left|h_{R_k D}\right|^2 P_{R_k}}{\sigma_D^2} \right) >0 \\
b &=& \frac{2\left|h_{S R_k}\right|^2}{\hat{\zeta} P_{R_k}^2} \left(2 + \frac{\left|h_{R_k D}\right|^2 P_{R_k}}{\sigma_D^2} \right) \\
c &=& -\frac{\left|h_{R_k D}\right|^2}{\sigma_D^2} \left(1 + \frac{\left|h_{R_k D}\right|^2 P_{R_k}}{\sigma_D^2} \right) < 0
\eeqn

We can easily see that the quadratic function $m(P_S) = a P_S^2 + b P_S + c$ has two roots, namely $\tilde{P}_{S1}$ and $\tilde{P}_{S2}$.
Because $c<0$ and $a>0$, $\Omega = b^2 - 4ac > 0$.
These quantities can be written as
\beqn
\tilde{P}_{S1} &=& \frac{-b-\sqrt{b^2 - 4ac}}{2a}\\
\tilde{P}_{S2} &=& \frac{-b+\sqrt{b^2 - 4ac}}{2a} \label{EQN_P_S_2_PROOF}
\eeqn
Note that $\tilde{P}_{S1}<0$ and $\tilde{P}_{S2}>0$ because 
\beqn
\frac{-b-\sqrt{b^2 - 4ac}}{2a} < \frac{-b-\sqrt{b^2}}{2a} = 0\\
\frac{-b+\sqrt{b^2 - 4ac}}{2a} > \frac{-b+\sqrt{b^2}}{2a} = 0
\eeqn
Hence $\mathcal{SC}_1$ can be rewritten as
\beqn
\label{EQN_SC_1_FINAL}
\mathcal{SC}_1 = \frac{\left|h_{R_k D}\right|^2 a}{P_S^6 \sigma_D^2}  \left(P_S - \tilde{P}_{S1}\right) \left(P_S - \tilde{P}_{S2}\right)
\eeqn
We now can find that if $0< P_S < \tilde{P}_S$, where $\tilde{P}_S = \tilde{P}_{S2}$ then $\mathcal{SC}_1 < 0$.
So we complete the proof.

\section{Proof of Lemma ~\ref{lem: nonco-convex}}
\label{Lemma1}

We can easily confirm that the constraints in \textbf{Problem 2} are convex sets due to the fact that all of them are linear functions. 
To prove \textbf{Problem 2} is convex optimization problem, we must prove that the function $\mathcal{\bar C}_k(P_S, P_{R_k})$ is concave.
Note that $\mathcal{\bar C}_k(P_S, P_{R_k})$ and $f\left(P_S, P_{R_k}\right)$ are from (\ref{EQN_C_bar}) and (\ref{EQN_g_FUNC}), respectively.

We can see that if both the terms in (\ref{EQN_C_bar}) are concave, then $\mathcal{\bar C}_k(P_S, P_{R_k})$ is concave.
Here the first term, $\frac{P_S \left|h_{SD}\right|^2}{\sigma_D^2}$ is concave.
So to prove that $\mathcal{\bar C}_k(P_S, P_{R_k})$ is concave, we can instead prove that $f\left(P_S, P_{R_k}\right)$ is a convex function \cite{Boyd04}.
For given $P_S \in \left[0, P_S^{\sf max}\right]$, we take the first-order partial derivative of $f\left(P_S, P_{R_k}\right)$ with respect to $P_{R_k}$ as 
\beqn
\frac{\partial f}{\partial P_{R_k}} = \frac{\left|h_{R_kD}\right|^2}{P_S \sigma_D^2} - \frac{\left|h_{SR_k}\right|^2}{\hat{\zeta} P_{R_k}^2} 
\eeqn
Then the second-order partial derivative of $f\left(P_S, P_{R_k}\right)$ with respect to $P_{R_k}$ can be determined as
\beqn
\frac{\partial^2 f}{\partial P_{R_k}^2} = \frac{2  \left|h_{S R_k}\right|^2}{\hat{\zeta} P_{R_k}^3}
\eeqn
We can see that $\frac{\partial^2 f}{\partial P_{R_k}^2} > 0$, hence $f\left(P_S, P_{R_k}\right)$ is a convex function for $P_{R_k}$.

Similarly, for given $P_{R_k} \in \left[0, P_{R_k}^{\sf max}\right]$, we take the first derivative of $f\left(P_S, P_{R_k}\right)$ with respect to $P_S$ as 
\beqn
\frac{\partial f}{\partial P_S} = \frac{-1}{P_S^2} - \frac{\left|h_{R_kD}\right|^2 P_{R_k}}{P_S^2 \sigma_D^2} 
\eeqn
Then the second derivative of $f\left(P_S, P_{R_k}\right)$ with respect to $P_S$ can be calculated as
\beqn
\frac{\partial^2 f}{\partial P_S^2} = \frac{2}{P_S^3} + \frac{2 \left|h_{R_k D}\right|^2 P_{R_k}}{P_S^3 \sigma_D^2}
\eeqn
Since $\frac{\partial^2 f}{\partial P_S^2} >0$, we can conclude that $f\left(P_S, P_{R_k}\right)$ is a convex function for $P_S$.
The proof of Lemma ~\ref{lem: nonco-convex} is complete.

\section{Proof of Lemma ~\ref{Lemma_noncoh1_zeta0}}
\label{Lemma01_zeta0}

We now consider the special case of ideal self-interference cancellation, i.e., $\hat{\zeta} = 0$.
Moreover the function $\log_2(1+x)$ is the strictly increase function of variable $x$, so we can rewrite the objective function of \textbf{Problem 1} as
\beqn
\mathcal{\tilde{C}}_k(P_S, P_{R_k}) =  \frac{\frac{P_{R_k} \left|h_{R_kD}\right|^2}{\sigma_D^2} \frac{P_S \left|h_{SR_k}\right|^2}{\sigma_{R_k}^2}}  {1+ \frac{P_{R_k} \left|h_{R_kD}\right|^2}{\sigma_D^2} + \frac{P_S \left|h_{SR_k}\right|^2}{\sigma_{R_k}^2}}
\eeqn
Here we approximate the $\mathcal{\tilde{C}}_k(P_S, P_{R_k})$ in the high SNR region which is usually used in wireless communications \cite{Tse05}.
So $\mathcal{\tilde{C}}_k(P_S, P_{R_k})$ is rewritten as
\beqn
\mathcal{\tilde{C}}_k(P_S, P_{R_k}) =  \frac{\frac{P_{R_k} \left|h_{R_kD}\right|^2}{\sigma_D^2} \frac{P_S \left|h_{SR_k}\right|^2}{\sigma_{R_k}^2}}  {\frac{P_{R_k} \left|h_{R_kD}\right|^2}{\sigma_D^2} + \frac{P_S \left|h_{SR_k}\right|^2}{\sigma_{R_k}^2}}
\eeqn

We can easily confirm that the constraints in \textbf{Problem 1} are convex sets due to the fact that all of them are linear functions. 
To prove \textbf{Problem 1} is the convex optimization problem, we must prove that the function $\mathcal{\tilde{C}}_k(P_S, P_{R_k})$ is concave. 
We rewrite $\mathcal{\tilde{C}}_k(P_S, P_{R_k})$ as follows:
\beqn\label{EQN_C_bar_zeta0}
\mathcal{\tilde{C}}_k(P_S, P_{R_k}) = \frac{\frac{\left|h_{R_kD}\right|^2}{\sigma_D^2} \frac{\left|h_{SR_k}\right|^2}{\sigma_{R_k}^2}}  {\tilde{f}\left(P_S, P_{R_k}\right)}
\eeqn
where
\beqn \label{EQN_g_FUNC_zeta0}
\tilde{f}\left(P_S, P_{R_k}\right) = \frac{\left|h_{R_kD}\right|^2}{P_S \sigma_D^2} + \frac{\left|h_{SR_k}\right|^2}{\sigma_{R_k}^2 P_{R_k}}
\eeqn

To prove that $\mathcal{\tilde{C}}_k(P_S, P_{R_k})$ is concave, we can instead prove that $\overline{f}\left(P_S, P_{R_k}\right) = 1/\tilde{f}\left(P_S, P_{R_k}\right)$ is a concave function.
To do so, we first determine the Hessian matrix of $\overline{f}\left(P_S, P_{R_k}\right)$ which can be expressed as
\beqn
\mathcal{H}_{\tilde{f}} = \left[\begin{array}{cc} 
\mathcal{\tilde{H}}_{11} & \mathcal{\tilde{H}}_{12} \\
\mathcal{\tilde{H}}_{21} & \mathcal{\tilde{H}}_{22}
\end{array}\right]
\eeqn

Here $\left\{\mathcal{H}_{ij}\right\}$, $i, j \in \left\{1,2\right\}$ are given as
\beqn
\label{EQN_H_11_PROOF_zeta0}
\begin{array}{llll} 
\mathcal{\tilde{H}}_{11} & = & \frac{\partial^2 \overline{f}}{\partial P_S^2} & = -\frac{\tilde{f} \frac{\partial^2 \tilde{f}}{\partial P_S^2} - 2 \left(\frac{\partial \tilde{f}}{\partial P_S}\right)^2}{\tilde{f}^3}  \label{EQN_H_11_PROOF_zeta0}\\
\mathcal{\tilde{H}}_{12} & = & \frac{\partial^2 \overline{f}}{\partial P_S \partial P_{R_k}}  & = -\frac{\tilde{f} \frac{\partial^2 \tilde{f}}{\partial P_S \partial P_{R_k}} - 2 \frac{\partial \tilde{f}}{\partial P_S} \frac{\partial \tilde{f}}{\partial P_{R_k}}}{\tilde{f}^3} \label{EQN_H_12_PROOF_zeta0} \\
\mathcal{\tilde{H}}_{21} & = & \frac{\partial^2 \overline{f}}{\partial P_{R_k} \partial P_S}  & = -\frac{\tilde{f} \frac{\partial^2 \tilde{f}}{\partial P_{R_k} \partial P_S} - 2 \frac{\partial \tilde{f}}{\partial P_S} \frac{\partial \tilde{f}}{\partial P_{R_k}}}{\tilde{f}^3} \label{EQN_H_21_PROOF_zeta0} \\
\mathcal{\tilde{H}}_{22} & = & \frac{\partial^2 \overline{\tilde{f}}}{\partial P_{R_k}^2}  & = -\frac{\tilde{f} \frac{\partial^2 \tilde{f}}{\partial P_{R_k}^2} - 2 \left(\frac{\partial \tilde{f}}{\partial P_{R_k}}\right)^2}{f^3}  \label{EQN_H_22_PROOF_zeta0}
\end{array}
\eeqn
where 
\beqn
\label{EQN_f_derivatives_zeta0}
\begin{array}{llll} 
\frac{\partial^2 \tilde{f}}{\partial P_S^2} & = & \frac{2\left|h_{R_k D}\right|^2}{\sigma_D^2 P_S^3} \\
\frac{\partial^2 \tilde{f}}{\partial P_S \partial P_{R_k}} & = & 0 \\
\frac{\partial^2 \tilde{f}}{\partial P_{R_k} \partial P_S} & = & 0 \\
\frac{\partial^2 \tilde{f}}{\partial P_{R_k}^2} & = & \frac{2  \left|h_{S R_k}\right|^2}{\sigma_{R_k}^2 P_{R_k}^3} \\
\frac{\partial \tilde{f}}{\partial P_{R_k}} & = & - \frac{\left|h_{SR_k}\right|^2}{\sigma_{R_k}^2 P_{R_k}^2} \\
\frac{\partial \tilde{f}}{\partial P_S} & = & - \frac{\left|h_{R_kD}\right|^2}{\sigma_D^2 P_S^2} 
\end{array}
\eeqn

Substitute (\ref{EQN_f_derivatives_zeta0}) to (\ref{EQN_H_11_PROOF_zeta0}), we obtain
\beqn
\begin{array}{llll} 
\mathcal{\tilde{H}}_{11} & = & \frac{-1}{\tilde{f}^3} \frac{2\left|h_{R_k D}\right|^2 \left|h_{SR_k}\right|^2}{\sigma_D^2 \sigma_{R_k}^2 P_S^3 P_{R_k}} < 0 \\
\mathcal{\tilde{H}}_{12} & = & \frac{2}{\tilde{f}^3} \frac{2\left|h_{R_k D}\right|^2 \left|h_{SR_k}\right|^2}{\sigma_D^2 \sigma_{R_k}^2 P_S^2 P_{R_k}^2} = \mathcal{\tilde{H}}_{21}  \\
\mathcal{\tilde{H}}_{22} & = & \frac{-1}{\tilde{f}^3} \frac{2\left|h_{R_k D}\right|^2 \left|h_{SR_k}\right|^2}{\sigma_D^2 \sigma_{R_k}^2 P_S P_{R_k}^3} 
\end{array}
\eeqn
We can see that $\mathcal{H}_{11} < 0$ and $\mathcal{H}_{11} \mathcal{H}_{22} - \mathcal{H}_{12} \mathcal{H}_{21} = 0$, i.e., the Hessian matrix $\mathcal{H}_{\tilde{f}}$ is negative semi-definite according to the Sylvester's criterion \cite{Horn12}.
Hence the function $\overline{f}\left(P_S, P_{R_k}\right)$ is the convex function.
Thus \textbf{Problem 1} is the convex optimization problem for the case of $\hat{\zeta} = 0$. %
So the proof of Lemma ~\ref{Lemma_noncoh1_zeta0} is completely done.

\section{Proof of Lemma ~\ref{Lemma_coh1}}
\label{Lemma02}

We now prove that \textbf{Problem 4} is not the strictly convex optimization problem by using contradiction.
To prove \textbf{Problem 4} is convex optimization problem, we must prove that the function $\mathcal{\breve{C}}^{\sf coh}_k(p_S, p_{R_k})$ is concave.
We rewrite $\mathcal{\breve{C}}^{\sf coh}_k(p_S, p_{R_k})$ as follows:
\beqn\label{EQN_C_breve}
\mathcal{\breve{C}}^{\sf coh}_k(p_S, p_{R_k}) = \frac{\frac{\left|h_{R_kD}\right|^2}{\sigma_D^2} \frac{\left|h_{SR_k}\right|^2}{\hat{\zeta}}}  {g\left(p_S, p_{R_k}\right)}
\eeqn
where
\beqn
\label{EQN_g_FUNC}
g\left(p_S, p_{R_k}\right) = \frac{1}{p_S^2}+ \frac{p_{R_k}^2 \left|h_{R_kD}\right|^2}{p_S^2 \sigma_D^2} + \frac{\left|h_{SR_k}\right|^2}{\hat{\zeta} p_{R_k}^2}
\eeqn

To prove that $\mathcal{\breve{C}}^{\sf coh}_k(p_S, p_{R_k})$ is concave, we can instead prove that $\overline{g}\left(p_S, p_{R_k}\right) = 1/g\left(p_S, p_{R_k}\right)$ is a concave function.
To do so, we first determine the Hessian matrix of $\overline{g}\left(P_S, P_{R_k}\right)$ which can be expressed as
\beqn
\mathcal{G}_g = \left[\begin{array}{cc} 
\mathcal{G}_{11} & \mathcal{G}_{12} \\
\mathcal{G}_{21} & \mathcal{G}_{22}
\end{array}\right]
\eeqn
Here $\left\{\mathcal{G}_{ij}\right\}$, $i, j \in \left\{1,2\right\}$ are expressed as follows:
\beqn
\label{EQN_G_11_PROOF}
\begin{array}{llll} 
\mathcal{G}_{11} & = & \frac{\partial^2 \overline{g}}{\partial p_S^2} & = -\frac{g \frac{\partial^2 g}{\partial p_S^2} - 2 \left(\frac{\partial g}{\partial p_S}\right)^2}{g^3}  \label{EQN_G_11_PROOF}\\
\mathcal{G}_{12} & = & \frac{\partial^2 \overline{g}}{\partial p_S \partial p_{R_k}} & = -\frac{g \frac{\partial^2 g}{\partial p_S \partial p_{R_k}} - 2 \frac{\partial g}{\partial p_S} \frac{\partial g}{\partial p_{R_k}}}{g^3} \label{EQN_G_12_PROOF} \\
\mathcal{G}_{21} & = & \frac{\partial^2 \overline{g}}{\partial p_{R_k} \partial p_S} & = -\frac{g \frac{\partial^2 g}{\partial p_{R_k} \partial p_S} - 2 \frac{\partial g}{\partial p_S} \frac{\partial g}{\partial p_{R_k}}}{g^3} \label{EQN_G_21_PROOF} \\
\mathcal{G}_{22} & = & \frac{\partial^2 \overline{g}}{\partial p_{R_k}^2} & = -\frac{g \frac{\partial^2 g}{\partial p_{R_k}^2} - 2 \left(\frac{\partial g}{\partial p_{R_k}}\right)^2}{g^3}  \label{EQN_G_22_PROOF}
\end{array}
\eeqn
where 
\beqn
\label{EQN_g_derivatives}
\begin{array}{llll} 
\frac{\partial^2 g}{\partial p_S^2} & = &\frac{6}{p_S^4} + \frac{6 \left|h_{R_k D}\right|^2 p_{R_k}^2} {\sigma_D^2 p_S^4}\\
\frac{\partial^2 g}{\partial p_S \partial p_{R_k}} & = &- \frac{4 p_{R_k} \left|h_{R_k D}\right|^2}{\sigma_D^2 p_S^3}\\
\frac{\partial^2 g}{\partial p_{R_k} \partial p_S} & = &- \frac{4 p_{R_k} \left|h_{R_k D}\right|^2}{\sigma_D^2 p_S^3} \\
\frac{\partial^2 g}{\partial p_{R_k}^2} & = &\frac{2  \left|h_{R_kD}\right|^2}{p^2_S \sigma_D^2} + \frac{6 \left|h_{SR_k}\right|^2}{\hat{\zeta} p_{R_k}^4}\\
\frac{\partial g}{\partial p_{R_k}} & = &\frac{2h_{R_kD}^2 p_{R_k}}{p_S^2 \sigma_D^2} - \frac{2 \left|h_{SR_k}\right|^2}{\hat{\zeta} p_{R_k}^3} \\
\frac{\partial g}{\partial p_S} & = &\frac{-2}{p_S^3} - \frac{2 \left|h_{R_kD}\right|^2 p_{R_k}^2}{p_S^3 \sigma_D^2} 
\end{array}
\eeqn
We can also observe that $\mathcal{G}_g$ is the symmetric matrix, i.e., $\mathcal{G}_{12} = \mathcal{G}_{21}$.
Because $\frac{\partial^2 g}{\partial p_S \partial p_{R_k}} = \frac{\partial^2 g}{\partial p_{R_k} \partial p_S} = - \frac{4 p_{R_k} \left|h_{R_k D}\right|^2}{\sigma_D^2 p_S^3}$.

According to the Sylvester's criterion \cite{Horn12}, the Hessian matrix $\mathcal{G}_g$ is negative definite iff $\mathcal{G}_{11} < 0$ and $\mathcal{G}_{11} \mathcal{G}_{22} - \mathcal{G}_{12} \mathcal{G}_{21} > 0$.
However, we can observe that $\mathcal{G}_{11} \mathcal{G}_{22} - \mathcal{G}_{12} \mathcal{G}_{21} < 0$ in the case that $p_{R_k} \in \left[0, \tilde{p}_{R_k}\right)$ and $p_S \in \left[0, \tilde{p}_S\right)$.
The proof and the values of $\left(\tilde{p}_{R_k}, \tilde{p}_S\right)$ are given in Appendix~\ref{Lemma02_1}.
Therefore the function $g\left(p_S, p_{R_k}\right)$ is not the strictly convex function.
Thus \textbf{Problem 4} is not the strictly convex optimization problem.
So we complete the proof of Lemma ~\ref{Lemma_coh1}.

\section{Proof of $\mathcal{G}_{11} \mathcal{G}_{22} - \mathcal{G}_{12} \mathcal{G}_{21} < 0$ when $p_{R_k} \in \left[0, \tilde{p}_{R_k}\right)$ and $p_S \in \left[0, \tilde{p}_S\right)$}
\label{Lemma02_1}

In this section, we prove that for any given $p_{R_k}$ ($p_{R_k} \in \left(\tilde{p}_{R_k}, \sqrt{P_{R_k}^{\sf max}}\right]$) we can choose the value of $p_S$ ($p_S \in \left[0, \tilde{p}_S\right)$) such that 
$\mathcal{G}_{11} \mathcal{G}_{22} - \mathcal{G}_{12} \mathcal{G}_{21} < 0$, where $\left(\tilde{p}_{R_k}, \tilde{p}_S\right)$ is from ( \ref{EQN_p_R_k_tilde_SUB}) and (\ref{EQN_p_S_SUB_lemma})).

We define $\mathcal{SC}_2$ as follows:
\beqn
\label{EQN_SC_2_Proof}
\mathcal{SC}_2 = \mathcal{G}_{11} \mathcal{G}_{22} - \mathcal{G}_{12} \mathcal{G}_{21} = \mathcal{G}_{11} \mathcal{G}_{22} - \left(\mathcal{G}_{12}\right)^2
\eeqn
Hence we must prove that $\mathcal{SC}_2<0$.
Moreover, from (\ref{EQN_SC_2_Proof}), we can instead prove that $\mathcal{G}_{11} \mathcal{G}_{22}<0$ which will be done in the following.

We substitute all parameters in (\ref{EQN_g_derivatives}) to $\mathcal{G}_{11}$ and $\mathcal{G}_{22}$ at (\ref{EQN_G_11_PROOF}).
After using some manipulations, we obtain
\beqn
\mathcal{G}_{11}\!\!  = \!\!\frac{2\left(\sigma_D^2 + \left|h_{R_k D}\right|^2 p_{R_k}^2\right)}{g^3 \hat{\zeta} p_{R_k}^2 p_S^6 \sigma_D^2}  \left(\hat{\zeta} p_{R_k}^2 - 6\left|h_{S R_k}\right|^2 p_S^2\right) \label{EQN_G11_SUB} \\
\mathcal{G}_{22}\!\!  = \frac{1}{g^3 p_S^4} \left(a_1 p_S^4 + 2 a_2 p_S^2 + a_3\right) \hspace{2.8cm} \label{EQN_G22_SUB}
\eeqn
where
\beqn
a_1 & = & \frac{2\left|h_{S R_k}\right|^2}{\hat{\zeta}^2 p_{R_k}^6} >0 \label{EQN_a1_SUB} \\
a_2 & = & \frac{-3 \left|h_{S R_k}\right|^2}{\hat{\zeta} p_{R_k}^4 \sigma_D^2} \left(\sigma_D^2 + 4 \left|h_{R_k D}\right|^2 p_{R_k}^2\right) <0 \label{EQN_a2_SUB} \\
a_3 & = & \frac{-2 \left|h_{R_k D}\right|^2}{\sigma_D^4} \left(\sigma_D^2 - 3 \left|h_{R_k D}\right|^2 p_{R_k}^2\right)  \label{EQN_a3_SUB} 
\eeqn

From (\ref{EQN_G11_SUB}), if we choose $p_S^2 < \tilde{p}_{S,1}^2$ , then $\mathcal{G}_{11} > 0$.
Here $\tilde{p}_{S,1}$ is written as
\beqn
\tilde{p}_{S,1} = \sqrt{\frac{\hat{\zeta}}{6}} \frac{p_{R_k}}{\left|h_{S R_k}\right|} \label{EQN_p_tilde_S_SUB1}
\eeqn
To complete the proof, we will find the ranges $\left(p_S,p_{R_k}\right)$ (where $p_S < \tilde{p}_{S,1}$) such that $\mathcal{G}_{22} < 0$.
Note that there are many such ranges $\left(p_S,p_{R_k}\right)$.
However we just give one example as follows.
Let consider the quadratic function $\eta (p_S) = a_1 p_S^4 + 2 a_2 p_S^2 + a_3$.
From (\ref{EQN_a3_SUB}), if $p_{R_k}^2 < \tilde{p}_{R_k}^2$, then $a_3 <0$.
Here $\tilde{p}_{R_k}$ is defined as
\beqn
\tilde{p}_{R_k} = \frac{\sigma_D}{\sqrt{3}\left|h_{R_k D}\right|} \label{EQN_p_R_k_tilde_SUB}
\eeqn

In this case, let us check the $\Omega$ as
\beqn
\Omega = a_2^2 - a_1 a_3 
\eeqn
We can observe that $\Omega > a_2^2 > 0$ (because $a_1>0$ and $a_3<0$).
So the quadratic function $\eta (p_S)$ has two real roots as follows:
\beqn
p_{S,1} = \frac{1}{a_1}\left(- a_2 -  \sqrt{\Omega}\right) < 0 \\
\tilde{p}_{S,2} = \frac{1}{a_1}\left(- a_2 +  \sqrt{\Omega}\right) > 0 \label{EQN_p_S_2_SUB}
\eeqn 
We should note that $\tilde{p}_{S,2} > 0$ because 
\beqn
\frac{1}{a_1}\left(- a_2 +  \sqrt{\Omega}\right) > \frac{1}{a_1}\left(- a_2 +  \sqrt{a_2^2}\right) = 0
\eeqn
Hence $\eta (p_S)$ can be written as follows:
\beqn
\eta (p_S) =  a_1 \left(p_S - p_{S,1}\right) \left(p_S - \tilde{p}_{S,2}\right)
\eeqn
It is clearly seen that if $0 \leq p_S < \tilde{p}_{S,2}$, then $\eta (p_S) < 0$ and hence $\mathcal{G}_{22} < 0$.

In summary, let us define $\tilde{p}_S$ as
\beqn
\tilde{p}_S = \min \left\{\tilde{p}_{S,1},\tilde{p}_{S,2}\right\} \label{EQN_p_S_SUB_lemma}
\eeqn
where $\tilde{p}_{S,1}$ and $\tilde{p}_{S,2}$ are from (\ref{EQN_p_tilde_S_SUB1}) and (\ref{EQN_p_S_2_SUB}), respectively.
We can conclude that if $0 \leq p_S < \tilde{p}_S$ and $p_{R_k} < \tilde{p}_{R_k}$ ($\tilde{p}_{R_k}$ is from (\ref{EQN_p_R_k_tilde_SUB})) then $\mathcal{G}_{11} > 0$ and $\mathcal{G}_{22} < 0$; and hence $\mathcal{G}_{11} \mathcal{G}_{22} < 0$, i.e., $\mathcal{SC}_2 < 0$.
So the proof is completely done.

\section{Proof of Lemma ~\ref{lem: co-convex}}
\label{Lemma2}

We can easily confirm that the constraints in \textbf{Problem 4} are convex sets.
Because the constraint $\mathcal{\bar{I}}^{\sf coh}_k\left(p_S, p_{R_k}\right) \leq \mathcal{\bar I}_P$ is the convex set for variable $p_S$ and $p_{R_k}$ (see the proof in Appendix~\ref{Lemma5}); and the all of remaining are linear functions. 
To prove \textbf{Problem 4} is the convex optimization problem, we must prove that the function $\mathcal{\breve{C}}^{\sf coh}_k(p_S, p_{R_k})$ is concave.

To prove that $\mathcal{\breve{C}}^{\sf coh}_k(p_S, p_{R_k})$ is concave, we can instead prove that $g\left(p_S, p_{R_k}\right)$ is a convex function \cite{Boyd04}.
Here $\mathcal{\breve{C}}^{\sf coh}_k(p_S, p_{R_k})$ and $g\left(p_S, p_{R_k}\right)$ are given from (\ref{EQN_C_breve}) and (\ref{EQN_g_FUNC}), respectively.
For given $p_S \in \left[0, \sqrt{P_S^{\sf max}}\right]$, we take the first-order partial derivative of $g\left(p_S, p_{R_k}\right)$ with respect to $p_{R_k}$ as 
\beqn
\frac{\partial g}{\partial p_{R_k}} = \frac{2\left|h_{R_kD}\right|^2 p_{R_k}}{p_S^2 \sigma_D^2} - \frac{2 \left|h_{SR_k}\right|^2}{\hat{\zeta} p_{R_k}^3} 
\eeqn
Then the second-order partial derivative of $g\left(p_S, p_{R_k}\right)$ with respect to $p_{R_k}$ can be determined as
\beqn
\frac{\partial^2 g}{\partial p_{R_k}^2} = \frac{2 \left|h_{R_kD}\right|^2}{p_S^2 \sigma_D^2} + \frac{6 \left|h_{SR_k}\right|^2}{\hat{\zeta} p_{R_k}^4} 
\eeqn
We can see that $\frac{\partial^2 g}{\partial p_{R_k}^2} > 0$, hence $g\left(p_S, p_{R_k}\right)$ is a convex function for $p_{R_k}$.

Similarly, for given $p_{R_k} \in \left[0, \sqrt{P_{R_k}^{\sf max}}\right]$, we take the first derivative of $g\left(p_S, p_{R_k}\right)$ with respect to $p_S$ as 
\beqn
\frac{\partial g}{\partial p_S} = \frac{-2}{p_S^3} - \frac{2 \left|h_{R_kD}\right|^2 p_{R_k}^2}{p_S^3 \sigma_D^2} 
\eeqn
Then the second derivative of $g\left(p_S, p_{R_k}\right)$ with respect to $p_S$ can be calculated as
\beqn
\frac{\partial^2 g}{\partial p_S^2} = \frac{6}{p_S^4} + \frac{6 \left|h_{R_kD}\right|^2 p_{R_k}^2}{p_S^4 \sigma_D^2} 
\eeqn
Since $\frac{\partial^2 g}{\partial p_S^2} >0$, we can conclude that $g\left(p_S, p_{R_k}\right)$ is a convex function for $p_S$.
So we complete the proof of Lemma ~\ref{lem: co-convex}.

\section{Proof that $\mathcal{\bar{I}}^{\sf coh}_k$ is a convex function for $p_S$ and $p_{R_k}$}
\label{Lemma5}

In this section, we will prove that $\mathcal{\bar{I}}^{\sf coh}_k$ is a convex function for $p_S$ and $p_{R_k}$.
\beqn
\mathcal{\bar{I}}^{\sf coh}_k\left(p_S, p_{R_k}, \phi_{\sf opt}\right) = \left|A+B e^{-j\phi}\right|^2
\eeqn
where 
\beqn
A &=& h_{SP} p_S + h_{R_kP} \sqrt{\zeta} p_{R_k} = \left|A\right| \angle{\phi_A} \\
B &=& \left[h_{SR_k} p_S + h_{R_kR_k} \sqrt{\zeta} p_{R_k} + \frac{\sigma_{R_k}}{\sqrt{2}} (1+j)\right] \nonumber\\
& &\times G_k h_{R_kP} \sqrt{P_{R_k}} = \left|B\right| \angle{\phi_B}
\eeqn

Let us define $D$ as follows:
\beqn
D = h_{SR_k} p_S + h_{R_kR_k} \sqrt{\zeta} p_{R_k} + \frac{\sigma_{R_k}}{\sqrt{2}} (1+j)
\eeqn
We calculate $\left|D\right|^2 = D D^*$.
After some manipulations, we obtain 
\beqn
\left|D\right|^2 = G_k^{-2} + L 
\eeqn
where 
\beqn
L/2 &=& p_S h_{SR_k}^R \sqrt{\zeta} p_{R_k} h_{R_kR_k}^R + p_S h_{SR_k}^R \sigma_{R_k}/\sqrt{2}  \nonumber\\
& & + \sqrt{\zeta} p_{R_k} h_{R_kR_k}^R \sigma_{R_k}/\sqrt{2} + p_S h_{SR_k}^I \sqrt{\zeta} p_{R_k} h_{R_kR_k}^I \nonumber\\
& & + p_S h_{SR_k}^I \sigma_{R_k}/\sqrt{2} + \sqrt{\zeta} p_{R_k} h_{R_kR_k}^I \sigma_{R_k}/\sqrt{2} \nonumber
\eeqn
Here the subscripts $R$ and $I$ (at $a^{R}$ and $a^{I}$) are the real and image of complex $a$.
Using the Cauchy–Schwarz Inequality \cite{Kadison52}, we have
\beqn
L/2 &\leq& p_S^2 \left[\left(h_{SR_k}^R\right)^2 + \left(h_{SR_k}^I\right)^2\right] \nonumber\\
& &+ \zeta p_{R_k}^2 \left[\left(h_{R_kR_k}^R\right)^2 + \left(h_{R_kR_k}^I\right)^2\right] + \sigma_{R_k}^2
\eeqn
So
\beqn
L/2 \leq p_S^2 \left|h_{SR_k}\right|^2 + \zeta p_{R_k}^2 \left|h_{R_kR_k}\right|^2 + \sigma_{R_k}^2 = G_k^{-2}
\eeqn

We use the approximation of $L \approx 2G_k^{-2}$ for the remaining proof.
So we have $\left|B\right|^2 = 3 p_{R_k}^2 \left|h_{R_kP}\right|^2$. 
Then we can rewrite $B$ as follows:
\beqn
B = \left|B\right| \angle{\phi_B} = \sqrt{3} p_{R_k} \left|h_{R_kP}\right| \angle{\phi_B}
\eeqn
The interference is now written as
\beqn
\mathcal{\bar{I}}^{\sf coh}_k\left(p_S, p_{R_k}\right) &=& \left|h_{SP} p_S + h_{R_kP} \sqrt{\zeta} p_{R_k} \right. \nonumber\\
& &\left.+ \sqrt{3} p_{R_k} \left|h_{R_kP}\right| \angle{\phi_B-\phi_{\sf opt}} \right|^2
\eeqn

We now prove that the constraint $\mathcal{\bar{I}}^{\sf coh}_k\left(p_S, p_{R_k}\right) \leq \mathcal{\bar I}_P$ is the convex set for variable $p_S$.
We can rewrite $\mathcal{\bar{I}}^{\sf coh}_k\left(p_S, p_{R_k}\right)$ as follows:
\beqn
\mathcal{\bar{I}}^{\sf coh}_k\left(p_S, p_{R_k}\right) = (h_{SP}^R p_S + F_1 p_{R_k})^2 + (h_{SP}^I p_S + F_2 p_{R_k})^2
\eeqn
where 
\beqn
F_1 &=& h_{R_kP}^R \sqrt{\zeta} + \sqrt{3} \left|h_{R_kP}\right| {\sf Cos} (\phi_B-\phi_{\sf opt})\\
F_2 &=& h_{R_kP}^I \sqrt{\zeta} + \sqrt{3} \left|h_{R_kP}\right| {\sf Sin} (\phi_B-\phi_{\sf opt})
\eeqn

We take the first-order and second-order partial derivatives of $\mathcal{\bar{I}}^{\sf coh}_k\left(p_S, p_{R_k}\right)$ with respect to $p_S$ as 
\beqn
\frac{\partial \mathcal{\bar{I}}^{\sf coh}_k}{\partial p_S} &=& 2 h_{SP}^R \left(h_{SP}^R p_S + F_1 p_{R_k}\right) \nonumber \\
& & + 2 h_{SP}^I \left(h_{SP}^I p_S + F_2 p_{R_k}\right)\\
\frac{\partial^2 \mathcal{\bar{I}}^{\sf coh}_k}{\partial p_S^2} &=& 2 \left(h_{SP}^R\right)^2 + 2 \left(h_{SP}^I\right)^2 = 2 \left|h_{SP}\right|^2\hspace{1.5cm}
\eeqn
Since $\frac{\partial^2 \mathcal{\bar{I}}^{\sf coh}_k}{\partial p_S^2} >0$, we can conclude that $\mathcal{\bar{I}}^{\sf coh}_k$ is a convex function for $p_S$.

Similarly, we prove that the constraint $\mathcal{\bar{I}}^{\sf coh}_k\left(p_S, p_{R_k}\right) \leq \mathcal{\bar I}_P$ is the convex set for variable $p_{R_k}$.
To do so, we take the first-order and second-order partial derivatives of $\mathcal{\bar{I}}^{\sf coh}_k\left(p_S, p_{R_k}\right)$ with respect to $p_{R_k}$ as 
\beqn
\frac{\partial \mathcal{\bar{I}}^{\sf coh}_k}{\partial p_{R_k}} &=& 2 F_1 (h_{SP}^R p_S + F_1 p_{R_k}) \nonumber\\
& & + 2 F_2 (h_{SP}^I p_S + F_2 p_{R_k})\\
\frac{\partial^2 \mathcal{\bar{I}}^{\sf coh}_k}{\partial p_{R_k}^2} &=& 2 F_1^2 + 2 F_2^2 
\eeqn
Since $\frac{\partial^2 \mathcal{\bar{I}}^{\sf coh}_k}{\partial p_{R_k}^2} >0$, we can conclude that $\mathcal{\bar{I}}^{\sf coh}_k$ is a convex function for $p_{R_k}$.
Hence we complete the proof.

\section{Proof of Lemma ~\ref{Lemma_coh2_zeta0}}
\label{Lemma02_zeta0}

We can easily confirm that the constraints in \textbf{Problem 4} are convex sets.
Because the constraint $\mathcal{\bar{I}}^{\sf coh}_k\left(p_S, p_{R_k}\right) \leq \mathcal{\bar I}_P$ is the convex set for variable $p_S$ and $p_{R_k}$ (see the proof in Appendix~\ref{Lemma5}); and the all of remaining are linear functions. 
To prove \textbf{Problem 4} is the convex optimization problem, we must prove that the function $\mathcal{\breve{C}}^{\sf coh}_k(p_S, p_{R_k})$ is a increased function for $p_S>0$ and $p_{R_k}>0$.
 
When $\hat{\zeta} = 0$, we can rewrite $\mathcal{\breve{C}}^{\sf coh}_k(p_S, p_{R_k})$ as
\beqn
\mathcal{\breve{C}}^{\sf coh}_k(p_S, p_{R_k}) = \frac{\left|h_{R_kD}\right|^2 p_{R_k}^2}{\sigma_D^2}.
\eeqn
We can easily observe that the quadratic function $\mathcal{\breve{C}}^{\sf coh}_k(p_S, p_{R_k})$ depends only on $p_{R_k}$ with the critical point of $p_{R_k}^* = 0$. 
Therefore $\mathcal{\breve{C}}^{\sf coh}_k(p_S, p_{R_k})$ is monotonically increasing in the range of $p_{R_k} > 0$.
Hence we can find the  globally optimal $p_{R_k}$ in its feasible range.
So we complete the proof of Lemma ~\ref{Lemma_coh2_zeta0}.

\bibliographystyle{IEEEtran}

\end{document}